\documentclass[a4paper]{lipics-v2021}

\title{Dynamic Membership for Regular Languages}
\author{Antoine Amarilli}{LTCI, Télécom Paris, Institut polytechnique de Paris, France \and \url{https://a3nm.net/}}{antoine.amarilli@telecom-paris.fr}{https://orcid.org/0000-0002-7977-4441}{}
\author{Louis Jachiet}{LTCI, Télécom Paris, Institut polytechnique de Paris, France}{}{}{}
\author{Charles Paperman}{Univ. Lille, CNRS, INRIA, Centrale Lille, UMR 9189 CRIStAL, F-59000 Lille, France}{}{https://orcid.org/0000-0002-6658-5238}{}
\authorrunning{A.\ Amarilli, L.\ Jachiet, C.\ Paperman}
\date{}

\hideLIPIcs

\nolinenumbers
\hypersetup{
    colorlinks,
    linkcolor={red!50!black},
    citecolor={blue!50!black},
    urlcolor={blue!30!black}
}

\Copyright{Antoine Amarilli, Louis Jachiet, and Charles Paperman}

\ccsdesc[100]{Theory of computation $\rightarrow$ Formal languages and automata theory}

\keywords{regular language, membership, RAM model, updates, dynamic}

\acknowledgements{We thank Jean-Éric Pin and Jorge Almeida for their advice
on~$\ZG$ and $\SG$, and thank the ICALP referees for their helpful feedback.
The authors have been partially supported by the ANR project EQUUS ANR-19-CE48-0019. Funded by the Deutsche Forschungsgemeinschaft (DFG, German Research Foundation) -- 431183758.
}

\newcommand{\card}[1]{\left|#1\right|}
\newcommand{\NN}{\mathbb{N}}
\newcommand{\ZZ}{\mathbb{Z}}

\newcommand{\calJ}{\mathcal{J}}
\newcommand{\eval}{\mathrm{eval}}
\newcommand{\V}{\mathbf{V}}
\newcommand{\A}{\mathbf{A}}

\newcommand{\Com}{\mathbf{Com}}

\newcommand{\MNil}{\mathbf{MNil}}

\newcommand{\QV}{\mathbf{QV}}

\newcommand{\SG}{\mathbf{SG}}
\newcommand{\ZG}{\mathbf{ZG}}
\newcommand{\ZE}{\mathbf{ZE}}

\newcommand{\LZG}{\mathbf{L}\ZG}

\newcommand{\LV}{\mathbf{L}\V}
\newcommand{\QLV}{\mathbf{QL}\V}
\newcommand{\QSG}{\mathbf{Q}\SG}
\newcommand{\LSG}{\mathbf{L}\SG}
\newcommand{\D}{\mathbf{D}}
\newcommand{\Nil}{\mathbf{Nil}}

\newcommand{\QLZG}{\mathbf{Q}\LZG}

\newcommand{\act}{\mathrm{act}}

\newcommand\restr[2]{{%
  \kern-\nulldelimiterspace %
  #1 %
  _{|#2} %
  }}

\usepackage[bibliography=common]{apxproof}
\usepackage{colonequals}
\usepackage{stmaryrd}
\usepackage{tabularx}
\usepackage{booktabs}

\theoremstyle{theorem}
\newtheoremrep{theorem}{Theorem}[section]
\newtheoremrep{proposition}[theorem]{Proposition}
\newtheoremrep{lemma}[theorem]{Lemma}
\newtheoremrep{claim}[theorem]{Claim}
\newtheoremrep{observation}[theorem]{Observation}
\newtheoremrep{conjecture}[theorem]{Conjecture}

\begin{document}

\maketitle
\begin{abstract}
  We study the \emph{dynamic membership problem} for regular languages: fix a
  language $L$, read a word~$w$, build in time $O(\card{w})$ a data structure
  indicating if $w$ is in~$L$, and maintain this structure efficiently
  under letter substitutions on~$w$. We consider this problem on the unit cost RAM model
  with logarithmic word length, where the problem always has a solution in
  $O(\log \card{w} / \log \log \card{w})$ per operation.

  We show that the problem is in $O(\log \log \card{w})$ for languages in an
  algebraically-defined, decidable class $\QSG$, and that it is in $O(1)$ for another
  such class~$\QLZG$. We show that languages not in $\QSG$ admit a reduction from the
  prefix problem for a cyclic group, so that they require $\Omega(\log \card{w} /\log
  \log \card{w})$ operations in the worst case; and that $\QSG$ languages not in
  $\QLZG$ admit a reduction from the prefix problem for the multiplicative
  monoid~$U_1 = \{0, 1\}$, which
  we conjecture cannot be maintained in~$O(1)$. This yields a conditional
  trichotomy. We also investigate intermediate cases between
  $O(1)$ and~$O(\log \log \card{w})$.
  
  Our results are shown via the dynamic word problem for monoids and semigroups,
  for which we also give a classification. We thus solve open problems of the paper of Skovbjerg Frandsen, Miltersen, and
  Skyum~\cite{skovbjerg1997dynamic} on the dynamic word problem, and
  additionally cover regular languages.
\end{abstract}

\section{Introduction}
\label{sec:intro}
This paper studies how to handle letter substitution updates on a word 
while maintaining the information of whether the word belongs to a regular language.
Specifically, we fix a regular language $L$ -- for instance $L = a^* b^*$.
We are then given an input word~$w$, e.g., $w = aaaa$.
We first preprocess $w$ in linear time to build a data structure, which we can
use in particular to test if $w \in L$. Now,
$w$ is edited by letter substitutions, and we want to update $w$ and keep track at each step of
whether $w \in L$. For instance, an update can replace the third letter of~$w$ by
a~$b$, so that $w = aaba$, which is no longer in~$L$. Then another update can
replace, e.g., the fourth letter of~$w$ by a~$b$, so that $w = aabb$, and now we
have $w \in L$ again. Our problem, called \emph{dynamic membership}, is to devise a
data structure to handle such update operations and determine 
whether~$w \in L$, as efficiently as possible. We study this task from a
theoretical angle, but it can also
be useful in practice to maintain a Boolean condition
(expressed as a regular language) on a user-edited word.

Dynamic membership was studied for various update operations, e.g.,
append operations for streaming algorithms or the sliding
window model~\cite{ganardi2016querying,ganardi2018automata,ganardi2019sliding},
letter substitutions for the dynamic word problem for
monoids~\cite{skovbjerg1997dynamic}, or concatenations and
splits~\cite{kirpichov2012incremental}. It was also studied in the 
case of \emph{pattern matching}, where 
we check if the word contains some target pattern~\cite{clifford2018upper}, which is also assumed
to be editable. It is also connected to
the incremental validation problem, which has been studied for strings and for
XML documents~\cite{balmin2004incremental}.
The problem was also studied from the angle of \emph{dynamic complexity}, which
does not restrict the running time but the logical language used to handle
updates~\cite{gelade2012dynamic}; and very recently refined to a study of the amount of
\emph{parallel work} required~\cite{schmidt2021work}.

Our focus in this work is to identify classes of fixed regular languages for which
dynamic membership can be solved extremely efficiently, e.g., in constant time
or sublogarithmic time. Our update language only
allows letter substitutions to the input word, in particular the length of the input
word can never be changed by updates. We make this choice because
insertions and deletions already
make it challenging to efficiently maintain the word itself (see
Section~\ref{sec:extensions}).
We work within the computational model of the unit-cost RAM, with logarithmic
cell size.

\subparagraph*{Dynamic word problem for
monoids~\cite{skovbjerg1997dynamic}.}
Our problem closely relates to the work
by Skovbjerg Frandsen, Miltersen, and
Skyum on the \emph{dynamic word problem for monoids}~\cite{skovbjerg1997dynamic}: fix a finite monoid,
read a \emph{word} which is a sequence of monoid
elements, and maintain under substitution updates the composition of these elements
according to the monoid's internal law. Indeed, the dynamic membership
problem for a language~$L$ reduces to the dynamic word problem for any monoid that
recognizes~$L$; but the converse is not true. Hence, studying the dynamic word
problem for monoids is coarser than studying the dynamic membership problem for
languages, although it is a natural first step and is already very challenging.

In the context of monoids, Skovbjerg Frandsen et al.~\cite{skovbjerg1997dynamic} 
propose a general algorithm for the dynamic word problem 
that can handle each operation in time $O(\log n / \log
\log n)$, for $n$ the length of the word. This is a refinement of the
elementary $O(\log n)$ algorithm that
decomposes the word as a balanced binary tree whose nodes are annotated with the
monoid image of the corresponding infix.
They show that the $O(\log n / \log \log n)$ bound is tight for some monoids, namely noncommutative
groups, and a generalization of them defined via an equation. This is obtained
by a reduction from the so-called \emph{prefix-$\ZZ_d$ problem}, for which an
$\Omega(\log n / \log \log n)$ lower bound~\cite{fredman1989cell} is known in
the cell probe model~\cite{green2021further}. We will reuse this lower bound in
our work.

They also show that the problem is easier for some monoids. 
For instance, commutative monoids can be maintained in $O(1)$, simply by
maintaining the number of element occurrences.
They also show a trickier $O(\log \log n)$ upper bound for
\emph{group-free monoids}:
this is based on a so-called Krohn-Rhodes decomposition~\cite{Rhodes68} and uses
a
predecessor data structure implemented as a van Emde Boas
tree~\cite{van1976design}.
However, there are non-commutative monoids for which the problem is in~$O(1)$ (as we will
show), and there is still a gap between group-free monoids (with an upper
bound in $O(\log \log n)$) and the monoids for which the 
$\Omega(\log n / \log \log n)$ lower bound applies. This was claimed as open
in~\cite{skovbjerg1997dynamic} and not addressed
afterwards. While there is a more recent study by Pǎtraşcu and
Tarniţǎ~\cite{patrascu2007dynamic}, it focuses on single-bit memory
cells.

\subparagraph*{Our contributions.}
In this paper, we attack these problems using algebraic monoid
theory.
This unlocks new results: first on the dynamic word problem for
monoids, where we extend the results of~\cite{skovbjerg1997dynamic}, and then on the
dynamic membership problem for regular languages.

We start with our results on the dynamic word problem for monoids, which are
summarized in Figure~\ref{fig:bestiaire} along with a table of the main classes
in Table~\ref{tab:classes}.
First, in Section~\ref{sec:sg}, we show how a more elaborate $O(\log \log n)$
algorithm can cover all monoids to which the $\Omega(\log n / \log \log n)$
lower bound of~\cite{skovbjerg1997dynamic} does not apply: we dub this class
$\SG$ and characterize it by the equation $x^{\omega+1} y
x^\omega = x^\omega y x^{\omega+1}$, for any elements $x$ and~$y$, where
$\omega$ denotes the idempotent power. Our algorithm shares some ideas with
the $O(\log \log n)$ algorithm of~\cite{skovbjerg1997dynamic}, in particular
it uses van Emde Boas trees, but it faces significant new challenges. For instance,
we can no longer use a Krohn-Rhodes decomposition, and proceed instead by a rather technical
induction on the $\mathcal{J}$-classes of the monoid.
Thus, we have an unconditional dichotomy between monoids in $\SG$,
which are in $O(\log \log n)$,
and monoids outside of $\SG$,
which are in $\Theta(\log n / \log \log n)$.

Second, in Section~\ref{sec:zg}, we generalize the $O(1)$ result on
commutative monoids to the monoid class~$\ZG$~\cite{auinger2000semigroups}.
This class is defined via the equation $x^{\omega+1} y = y x^{\omega+1}$,
i.e., only the elements that are part of a group are required to commute with
all other elements. We show that the dynamic word problem for these monoids is
in~$O(1)$: we use an algebraic characterization to reduce them to 
commutative monoids and to monoids obtained from nilpotent semigroups, for which
we design a simple but somewhat surprising algorithm.
We also show
a conditional lower bound: for any monoid~$M$ not in~$\ZG$, we can reduce the
\emph{prefix-$U_1$} problem to the dynamic word problem for~$M$. This is the
problem of maintaining a binary word under letter substitution updates while answering
queries asking if a prefix contains a~$0$. It can be seen as 
a priority queue (slightly weakened), so we conjecture that no $O(1)$ data structure for this problem
exists in the RAM model.
If this conjecture holds, $\ZG$ is exactly the class of monoids having a dynamic
word problem in~$O(1)$.

We then extend our results in Section~\ref{sec:semigroup} from monoids to the
dynamic word problem for semigroups. Our results for~$\SG$ extend directly:
the upper bound on~$\SG$ also applies to semigroups in~$\SG$, and semigroups
not in~$\SG$ must contain a submonoid not in~$\SG$ so covered by the lower
bound.
For $\ZG$, there are major complications, and we must study the class $\LZG$
of semigroups where all submonoids are in~$\ZG$. Semigroups not in~$\LZG$ are
covered by our conditional lower bound on prefix-$U_1$, but it is very tricky to show the
converse, i.e., that
imposing the condition on~$\LZG$ suffices to ensure tractability. We do so by
showing tractability for $\ZG*\D$, the semigroups generated by \emph{semidirect
products} of $\ZG$ semigroups and so-called \emph{definite semigroups}, and
by showing in~\cite{localityarxiv} that $\ZG*\D=\LZG$, a \emph{locality result}
of possible independent interest.

Next, we extend our results in Section~\ref{sec:languages} from semigroups to
languages. This is done using the notion of \emph{stable
semigroup}~\cite{Barrington92,Chaubard06}, denoted as the $\mathbf{Q}$ operator; 
and specifically the class $\QSG$ of regular languages where the stable
semigroup of the syntactic morphism in is~$\SG$, and the class $\QLZG$ where
all local monoids of the stable semigroup of the syntactic morphism are
in~$\ZG$. We obtain:

\begin{theorem}
  \label{thm:mainres}
Let $L$ be a regular language, and consider the dynamic membership problem
  for~$L$ on the unit-cost RAM with logarithmic word length under letter substitution
  updates:
\begin{itemize}
\item If $L$ is in the class $\QLZG$,
  then the problem is in~$O(1)$.
\item If $L$ is not in the class $\QLZG$ but is in the class $\QSG$, then the
  dynamic membership problem is in~$O(\log \log n)$ with $n$ the length of the
    word. Further, solving it in $O(1)$ time gives an $O(1)$
    implementation of a structure for the prefix-$U_1$ problem.
\item If $L$ is not in the class $\QSG$, then the dynamic membership problem is
  in $\Theta(\log n /\log \log n)$.
\end{itemize}
\end{theorem}

We last present in Section~\ref{sec:extensions} some extensions and questions
for future work: preliminary observations on the precise complexity of languages
in $\QSG\setminus\QLZG$ (as the $O(\log \log n)$ bound is not shown to be tight), the complexity of deciding which case of the theorem
applies, the support for insertion and deletion updates, and the support for
infix queries.

\section{Preliminaries and Problem Statement}
\label{sec:prelim}
\subparagraph*{Computation model.}
We work in the RAM model with unit cost, i.e., each cell
can store integers of value at most polynomial in
$O(\card{w})$ where $\card{w}$ is the length of the input,
and arithmetic operations (addition, successor, modulo, etc.) on two
cells take unit time.
As the integers have at most polynomial value, the memory usage is
also polynomially bounded.

We consider \emph{dynamic problems} where we are given an input word,
preprocess it in linear time to build a data structure,
and must then handle \emph{update operations} on the word (by reflecting them in
the data structure), and 
\emph{query operations} on the current state of the word (using the data
structure).
The \emph{complexity} of the problem is the worst-case
complexity of handling an update or answering a query.

Like in~\cite{skovbjerg1997dynamic}, the lower bounds
that we show actually hold in the coarser \emph{cell probe}
model, which only considers the number of memory cells accessed during a
computation. Furthermore, they hold even without the assumption that the
preprocessing is linear.

Given two dynamic problems $A$ and $B$, we say that $A$ has a
\emph{(constant-time) reduction} to $B$ if we can implement a data structure
for problem~$A$ having constant-time complexity when using as oracle constantly
many data structures for problem~$B$ (built during the preprocessing).
In other words, queries and updates on the structure for~$A$ can perform constant-time
computations using its own memory, but they can also use the data structures for~$B$ as
an oracle, i.e., perform a constant number of queries and updates on them,
which are considered to run in~$O(1)$. We similarly talk of 
a dynamic problem having a \emph{(constant-time) reduction} to multiple problems, meaning
we can use all of them as oracle.
If problem $A$ reduces to problems $B_1,
\ldots, B_n$, then any complexity upper bound that holds on all problems $B_1,
\ldots, B_n$ also holds for~$A$, and any complexity lower bound on~$A$ extends
to at least one of the~$B_i$.

\subparagraph*{Problem statement.}
Our problems require some algebraic prerequisites. We refer the
reader to the book of Pin~\cite{Pin84} and his lecture
notes~\cite{mpri} for more details.
A \emph{semigroup} is a set $S$ equipped with an associative composition law
(written multiplicatively), and a \emph{monoid} is a semigroup $M$ with a
\emph{neutral element}, i.e., an element $1$ such that $1x = x1 = x$ for all $x
\in M$; the neutral element is then unique. One example of a monoid is the
\emph{free monoid} $\Sigma^*$ defined for a finite alphabet $\Sigma$ and
consisting of all words with letters in~$\Sigma$ (including the empty word),
with concatenation as its law.
Except for the free monoid, all semigroups and monoids considered are finite.

A semigroup element $x \in S$ is \emph{idempotent} if $xx = x$.
For $x \in S$, we denote by $\omega$ the idempotent power of~$x$, i.e., the
least
positive integer such that $x^{\omega}$ is idempotent.
A \emph{zero} for~$S$ is an element $0 \in S$ such that $0x = x0 = 0$ for all $x \in S$: if it exists, it is unique.
Given a semigroup $S$, we write $S^1$ for the monoid obtained by adding a fresh
neutral element to~$S$ if it does not already have one.

A \emph{morphism} from a semigroup $S$ to a semigroup $S'$ is a map
$\mu\colon S\to S'$ such that for any $x, y \in S$, we have $\mu(xy) = \mu(x)
\mu(y)$. A \emph{morphism} from a monoid $M$ to a monoid $M'$ must additionally
verify that $\mu(1) = 1'$, for $1$ and $1'$ the neutral elements of~$M$ and
$M'$ respectively.

The \emph{direct product} of two monoids $M_1$ and $M_2$ is $M_1 \times M_2$
with componentwise composition; it is also a monoid. A \emph{quotient} of a
monoid $M$ is a monoid $M'$ such that there is a surjective morphism from~$M$
to~$M'$. A \emph{submonoid} is a subset of a monoid which is also a monoid.
The analogous notions for semigroups are defined in the expected way.
A \emph{pseudovariety} of monoids (resp., semigroups) is a class of monoids
(resp., semigroups) closed under 
direct product, quotient and
submonoid (resp., subsemigroup).
The pseudovariety of monoids (resp., semigroups) \emph{generated} by a class~$\V$ of
monoids (resp., of semigroups)
is the least pseudovariety closed under these operations and containing~$\V$.
As we are working with finite semigroups and monoids,
we refer to pseudovarieties simply as varieties.

We consider dynamic problems where we maintain a word $w$ on a finite
alphabet~$\Sigma$, every letter being stored in a cell. We allow
\emph{letter substitution} updates of the form $(i, a)$ for $1 \leq i \leq \card{w}$ and
$a \in \Sigma$. The letter substitution update $(i, a)$ replaces the $i$-th letter
of~$w$ by~$a$;
the size $\card{w}$ of the word never changes.
Given the input word $w$, we first preprocess it in time $O(\card{w})$ to build
a data structure. The data structure must then support update operations for
letter substitution updates, and some query operations (to be defined below). 
The complexity that we measure is the worst-case complexity 
of an update operation or query operation, as a function
of~$\card{w}$.
Our definition does not limit the memory used.
However, all
our upper complexity bounds 
actually have memory usage in~$O(\card{w})$.
Further, all our lower bounds hold even when no assumption is made on the memory
usage.

We focus on three related dynamic problems, allowing different query operations.
The first is the
\emph{dynamic word problem for monoids}: fix a monoid~$M$, the alphabet
$\Sigma$ is~$M$, and the query returns the \emph{evaluation} of the current
word~$w$, i.e., the product of the elements of~$w$ (it is an element of~$M$).
This is the problem studied in~\cite{skovbjerg1997dynamic}.
The second is the \emph{dynamic word problem for semigroups}, which is the same but
with a semigroup, and assuming that $\card{w}>0$. The third is the
\emph{dynamic membership problem for regular languages}: we fix a regular
language~$L$ on the alphabet~$\Sigma$, 
and the query checks whether the current word belongs to~$L$.

We study the data complexity of these problems in the rest of this paper,
i.e., the complexity expressed as a function of~$w$, with the monoid, semigroup,
or language being fixed.
Let us first observe that, for monoids and more generally for semigroups,
the usual algebraic operators of quotient,
subsemigroup, and direct product, do not increase the complexity of the problem:
\begin{toappendix}
\subsection{Proof of Proposition~\ref{prp:closure}}
\end{toappendix}
\begin{propositionrep}
  \label{prp:closure}
	Let $S$ and $T$ be finite semigroups. The
        dynamic word problem of subsemigroups or quotients of~$S$ reduces
        to the same problem for~$S$, and 
        the dynamic word problem of $S\times T$ reduces to the same problem for~$S$ and~$T$.
\end{propositionrep}
\begin{proof}
In the case of submonoids, we can simply solve the dynamic word problem
for the submonoid by performing the computations in $S$. In the case of
  quotient, we can represent each element of the quotient by some choice of
  representative element in~$S$, perform the computation in~$S$,
  and check in which equivalence class in the quotient does the result in~$S$
  fall.
For the product, we can simply maintain the structures for both $S$ and $T$ simultaneously.
\end{proof}

\subparagraph*{Hard problems.}
All our lower bounds are obtained by reducing from the problem \emph{prefix-$M$}, for
$M$ a fixed monoid. In this problem, we maintain a word of~$M^*$ under
letter substitution updates, and handle \emph{prefix queries}: given a
prefix length, return the evaluation of the prefix of that length. 

In particular, for $d \geq 2$, we consider the problem \emph{prefix-$\ZZ_d$} for
$\ZZ_d$ the cyclic group of order~$d$, i.e., $\ZZ_d = \{0, \ldots, d-1\}$ with
addition modulo~$d$, where the evaluation of prefix is the sum of its elements modulo~$d$. The
following lower bound is known already in the cell probe model:

\begin{theorem}[\cite{fredman1989cell,skovbjerg1997dynamic}]
  \label{thm:zzd}
  For any fixed $d \geq 2$,
  any structure for prefix-$\ZZ_d$ on a word of
  length~$n$ has complexity $\Omega(\log n / \log \log n)$.
\end{theorem}

We also consider the problem \emph{prefix-$U_1$}, where $U_1
= \{0, 1\}$ is the multiplicative monoid whose composition is the logical AND, i.e., prefix queries
check if the prefix contains an occurrence of~$0$. Equivalently, 
we must maintain a subset $S$ of a universe $\{1, \ldots,
n\}$ (intuitively $n$ is the length of the word) under insertions and
deletions, and support \emph{threshold queries} that ask, given
$0 \leq i \leq n$, whether $S$ contains some element which is $\leq i$
(i.e., if some position before~$i$ has a~$0$).
The prefix-$U_1$ problem can be solved
in $O(\log \log n)$~\cite{thorup2007equivalence}
with a priority queue data structure,
or even in expected $O(\sqrt{\log \log n})$ if we allow
randomization~\cite{han2002integer}.
Note that prefix-$U_1$
is slightly weaker than a priority queue as we can only \emph{compare} the
minimal value to a value given as input.
We do not know of lower bounds on prefix-$U_1$, but 
conjecture~\cite{cstheorymin} that it cannot be solved in~$O(1)$:

\begin{conjecture}
  \label{con:cpq}
  There is no structure for prefix-$U_1$ with complexity~$O(1)$.
\end{conjecture}
Note that the best algorithm for prefix-$U_1$ works by sorting small sets of large integers. This takes linear time in the cell
probe model, so does not rule out the existence of an $O(1)$
priority queue~\cite{thorup2007equivalence}. Hence, a lower bound for
prefix-$U_1$ would
need to apply to the RAM model specifically, which would require new
techniques.

Our last hard problem is prefix-${U_2}$
where $U_2$ is the monoid $\{1, a, b\}$ with composition law $xy=y$ for
$x,y\in\{a, b\}$,
i.e., queries check if the last non-neutral element is $a$ or~$b$ (or
nothing). By adapting known results on
the \emph{colored predecessor problem}~\cite{patrascu07}, we have:

\begin{toappendix}
\subsection{Proof of Theorem~\ref{thm:prefixU2pred}}
\end{toappendix}
\begin{theoremrep}[Adapted from \cite{patrascu07}]
\label{thm:prefixU2pred}
Any structure for prefix-${U_2}$ on a word of
length $n$ must be in $\Omega(\log \log n)$.
\end{theoremrep}

\begin{toappendix}
  To prove this result, we first introduce prerequisites on the colored
  predecessor problem. We then explain how it relates to the prefix-$U_2$
  problem. We then derive a lower bound for prefix-$U_2$.
\subparagraph*{Colored predecessor problem.}
The colored predecessor problem is a problem parameterized by $n$. In
this problem, we have a set $S\subseteq \{ 1, \ldots, n\}$ and a color
(black or white) for each element in $S$ and we need to answer queries
asking for the color $c_y$ of the biggest element $y\in S$ such that
$y\leq x$ for $x$ a parameter of the query.

In the static version of the problem, the set $S$ is given in advance
and we are allowed to compute an index before receiving the query. The
time complexity is then measured only in terms of the time required to
answer the query.

In~\cite{patrascu07}, it has been proved that, in the worst case, the
  static colored predecessor problem for $S$ a subset of $\{1, \ldots,
  n^2\}$ with $|S|=n$ cannot be solved in less than $\Omega(\log \log
  n)$ query time if the space allowed is bounded by $n^{1+o(1)}$. And
  this lower bound holds in the cell probe model, and still holds even
  if we allow randomization and allow the answer to be only probably correct.

\subparagraph*{Reduction from the static colored predecessor problem to prefix-$U_2$.}
Now let us suppose that we have a maintenance scheme for prefix-$U_2$
  with complexity $d(n)$. Let $S\subseteq \{1, \ldots, n^2\}$ be an instance of the
static colored predecessor problem and let us show that we can build a
data structure for $S$ that takes $O(n\times d(n^2))$ space and can
answer predecessor queries with $O(d(n^2))$ cell probes. If we manage
to do that then either $d(n^2) = n^{o(1)}$ and thus we have $d(n^2)
= \Omega(\log\log n)$ using the static lower bound or $d(n^2)\neq
n^{o(1)}$. In both cases we have $d(n^2) = \Omega(\log\log n)$ and
thus $d(n) = \Omega(\log\log n)$.

\subparagraph*{Lower bound for prefix-$U_2$.}
  Since $S\subseteq \{1,\ldots,n^2\}$, we create a word $w$ of length $n^2$ where
all elements are the neutral element of $U_2$. Then we perform $n$
letter substitution updates on $w$ so that the $i$-th letter of $w$ is an $a$
when $i\in S$ and $i$ has color White, is a $b$ when $i\in S$ and $i$ has
color Black, and is the neutral element of $U_2$ when $i\not\in S$. The
index for prefix-$U_2$ on $w$ can now be used to answer predecessor
queries on $S$, because querying the prefix of length $k$ returns the color of the
predecessor query with parameter $k$ (White when the result is $a$, Black when
it is $b$, and the neutral element when $k$ has no predecessor in $S$).

Since queries on this index take $d(n^2)$ time, we do have a scheme for
the static predecessor problem in time $d(n^2)$, but, before concluding,
we need to make sure that the space usage is bounded by
$O(n^{1+o(1)})$. This is not obvious for two reasons: first because
the initialization of the dynamic algorithm might modify $O(n^2)$
memory cells during the initialization (or even more if we allow a
superlinear preprocessing time) and second because during each update
the dynamic algorithm might access a memory cell $i$ with $i\neq
n^{1+o(1)}$ (e.g. it might use the memory cell $n^2$). We now explain how to
  handle these two issues.

\subparagraph*{Reducing the memory footprint.}
To handle the first issue, we can notice that the whole preprocessing
computation depends only on $n$ and therefore if a cell has not been
modified since the initialization of the dynamic algorithm then we can recover its
value from $n$. For the second issue, we can notice that each update
takes $d(n^2)$ time, and therefore the total number of memory cells
that have been modified during one of the updates is bounded by $d(n^2)\times
n$. Therefore using a perfect hashing scheme~\cite{perfectHashing}
using $O(d(n^2)\times n)$ memory cells, we can retrieve in $O(1)$ for
each address $i$ whether $i$ was modified during an update and if it
has been modified, what is its current value.

All in all, we can modify the query function for our dynamic problem
in the following way: whenever the dynamic algorithm wants to read the
cell at address $i$ we use the perfect hashing scheme to determine if
$i$ has been modified since the initialization, if it has been
modified we retrieve its value and otherwise we recompute the value
that can be deduced from $n$. Such an algorithm makes $O(d(n^2))$ cell
probes over a memory of size $O(n\times d(n^2))$. This is what we wanted to
obtain, concluding the proof.
\end{toappendix}

\subparagraph*{General algorithms.}
Of course, the ``hard'' prefix problems, and the three problems that we study,
can all be solved in
$O(\card{w})$ by re-reading the whole word
at each update.
We can improve this to $O(\log \card{w})$
by maintaining a balanced binary tree on the letters of the word, with each node of
the tree carrying
the evaluation in the monoid of the letters reachable from that node.
Any letter substitution update on the word can be propagated up to the root in logarithmic time, and
the annotation of the root is the evaluation of the word.
This algorithm has been implemented in practice~\cite{kirpichov2012incremental}.
A finer bound is given in~\cite{skovbjerg1997dynamic} using a
folklore technique of working with $(\log n)$-ary trees rather than binary trees,
and using the power of the RAM model. We recall it here for monoids (but it
applies to all three problems):

\begin{toappendix}
  \subsection{Proof of Theorem~\ref{thm:ub_ln_lln}}
\end{toappendix}
\begin{theoremrep}[\cite{skovbjerg1997dynamic}]
  \label{thm:ub_ln_lln}
  For any fixed monoid~$M$, 
  the dynamic word problem and prefix problem for~$M$
  are in $O(\log n / \log \log n)$.
\end{theoremrep}

\begin{toappendix}
  We first show the claim for the dynamic word problem.
  Let $(M,\odot)$ be the fixed monoid and $w = w_1 \ldots w_n $ be a word. Let
  us show that we can maintain the value $\nu(w) =
  w_1 \odot \cdots \odot w_n$ after letter substitution updates on $w$ in
  $O(\log n/ \log \log n)$.

  We denote by $G_M$ the (infinite) directed graph whose nodes are the
  elements of $M^*$ and where a node $(m_1, \dots, m_\ell)$ has
  $\ell\times \card{M}$ outgoing edges each labeled with a different
  element of $\{1, \ldots, \ell\}\times M$. For a node $(m_1, \dots, m_\ell)$, the
  edge labeled $(i,v)$ goes to $(m_1, \dots, m_{i-1}, v, \allowbreak m_{i+1},
  \allowbreak
  \ldots, m_\ell)$. Finally the \emph{value} of the node $(m_1, \dots,
  m_\ell)$ is defined as $m_1 \odot \dots \odot m_\ell$.

  We denote by $G_M^k$ the restriction of $G_M$ to nodes $(m_1, \dots,
  m_\ell)$ where $\ell\leq k$. $G_M^k$ has at most $\card{M}^{k}\times k$
  nodes and the degree of each node is bounded by $|M|\times k$ so
  there exists $\gamma$ (independent of $n$) such that for $k=\lfloor\gamma
  \times \log n\rfloor$, the graph $G^k_M$ can be computed in $O(n^{1/2})$ time and stored
  with $O(n^{1/2})$ space (in fact, for any $\epsilon>0$,
  we could achieve $n^{\epsilon}$ by changing $\gamma$).
  We can also ensure
  that the value of each node is pre-computed (and takes constant time
  to access) and that retrieving the $(i,v)$-labeled neighbor of a
  given node takes constant time.

  \begin{claim}
    Once $G_M^{k}$ is computed, we have a scheme $S_k(n)$
    with $O(\log_{k}(n))$ update time that allows to maintain the value
    of $\nu(w)$ after letter substitution updates for $w$, where $n$ is the length
    of~$w$.
  \end{claim}

  \begin{proof} This idea is based on a modification of the Fenwick
    tree data structure with a branching factor of $k$. The proof
    works by induction. Each scheme $S_k(n)$ will maintain a node in
    $G_M^k$ such that the product of elements of the word is equal to the
    value of this node (which takes $O(1)$ to retrieve).

\begin{itemize}

\item For $0 \leq n \leq k$, we will store the node $(c_1, \dots,
 c_n)$ where $c_1 \dots c_n$ is the word to maintain. We deal with
  letter substitutions using the transitions of the graph: to update the
  element at position $j$ to the element $v$ we simply replace the
  current node with its $(j,v)$-neighbor.

\item For $n>k$, we denote by $R$ the biggest power of $k$ that is strictly
  less than $n$ and we split the word into $w_0, \dots, w_{\ell}$ such
    that $|w_0| = \dots = |w_{\ell}-1| = R$ and $w_{\ell} \leq R$
    (note that $\ell\leq k$). At
  each step the current node will correspond to
  $(\nu(w_0), \dots, \nu(w_{\ell}))$. For each of the words
  $w_0, \dots, w_{\ell}$ we use a scheme $S_k(R)$ to maintain the
  value $\nu(w_i)$.

  To update $w$ at position $j$, we update the subword $w_i$ with
  $i=\left\lfloor \dfrac{j}{R} \right\rfloor$ at position $j-i\times R$ which
  gives us the new value of $\nu(w_i)$ then we update the current node
  by replacing it with its $(\nu(w_i),i)$-neighbor.\qedhere
\end{itemize}
  \end{proof}
  Notice that initialization time is linear, and each update will make
  $\log_k(n)$ recursive calls each in constant time. To finish, let us recall that,
  for $k=\gamma \times \log n$, the computation of $G_M^k$
  runs $O(n^{1/2})$, and each query and update operation runs in time
  $O(\log_k(n))=O\left(\dfrac{\log n}{\log(\gamma\times\log n)}\right)=O\left(\dfrac{\log n}{\log\log n}\right)$.
  This concludes the proof of Theorem~\ref{thm:ub_ln_lln}.

  In terms of memory usage, the data structure uses $O(n/k)$ that is $O(n/\log
  n)$. Note that this is sublinear thanks to the power of the RAM model. This is
  because nodes at the lowest level compress a factor of size $k$ of the
  original word. In particular, this memory usage is no more than linear in the
  size of the original word.

  We now turn to the prefix problem, and show that the structure can be used
  more generally to support prefix queries. In fact, we will show more generally
  that the structure can handle \emph{infix} queries
  (see Section~\ref{sec:extensions}).
  By lowering the $\gamma$ parameter, we can precompute for each node
$(m_1,\dots, m_\ell)$ of $G_M^k$ and for each $1\leq i \leq j \leq
\ell$ the value $m_{i} \odot \dots \odot m_{j}$. Then, in a scheme
$S_k(n)$, to compute the element of the monoid corresponding to the
factor $(i,j)$, either positions $i$ and $j$ belong to the same
subword (in which case we recurse on this subword), or the factor $(i,j)$
can be decomposed into a strict suffix of some subword $w_{i'}$
followed by a possibly empty contiguous list of subwords followed by a
strict prefix of a subword $w_{j'}$. If they are not empty, we recurse
on $w_{i'}$ and $w_{j'}$ to get the value of the strict suffix and
prefix parts and using our precomputation we get the value for the
list of subwords. By composing the obtained values, we get the element of the
monoid corresponding to the factor $(i,j)$.

Note that an infix query on a scheme might trigger two recursive
calls, but this can only happen if the infix query is not a prefix
nor a suffix query and the two recursive calls it makes are a prefix
query and a suffix query for their respective schemes. Therefore the overall algorithm does have
the expected $O\left(\dfrac{\log n}{\log\log n}\right)$ complexity.
\end{toappendix}

Our goal in this paper is to solve
the dynamic word problem and dynamic membership problem more efficiently for
specific classes of monoids, semigroups, and languages. We start our study with
monoids in the next two sections, by studying the varieties $\SG$ and $\ZG$.

\section{Dynamic Word Problem for Monoids in $\SG$}
\label{sec:sg}
We define the class $\SG$ of monoids by the equation
$x^{\omega+1} y x^\omega = x^\omega y x^{\omega+1}$ for all $x, y$.
It incidentally occurs
in~\cite[Theorem~3.1]{azevedo1990join}, but to our knowledge was not otherwise
studied.
The name $\SG$ means \emph{swappable groups}.
Intuitively, a monoid $M$ is in $\SG$
iff, for any two elements $r, t \in M$ belonging to the same subgroup of~$M$,
we can \emph{swap} them, i.e., 
$r s t = t s r$ for all $s \in M$.
We first recall the lower bound from~\cite{skovbjerg1997dynamic} on the dynamic
word problem for monoids \emph{not in} $\SG$, and then show an upper bound for monoids
in~$\SG$.

\subparagraph*{Lower bound.}
The monoids not in~$\SG$ are in fact those covered by the lower
bound of~\cite{skovbjerg1997dynamic}. Namely, we have the
following, implying the $\Omega(\log n / \log \log n)$ lower bound  by
Theorem~\ref{thm:zzd}:

\begin{theorem}[\cite{skovbjerg1997dynamic}, Theorem~2.5.1]
  \label{thm:lbsg}
  For any monoid $M$ not in~$\SG$, there exists $d \geq 2$ such that 
  the prefix-$\ZZ_d$ problem reduces to the dynamic word
  problem for~$M$.
\end{theorem}

\subparagraph*{Upper bound.}
The rest of this section presents our upper bound on monoids in~$\SG$.
In fact, we show a more general claim on the dynamic word 
problem for \emph{semigroups} in~$\SG$, i.e., those satisfying the equation
of~$\SG$. This covers in particular the monoids
of~$\SG$:

\begin{theorem}
  \label{thm:ubsg}
  The dynamic word problem for any semigroup in~$\SG$ is in $O(\log \log n)$.
\end{theorem}

This result extends the result of \cite{skovbjerg1997dynamic} on group-free
monoids, because  $\SG$ contains all aperiodic monoids and all commutative monoids.
Indeed, an aperiodic monoid satisfies the equation
$x^{\omega+1}=x^\omega$, and then
$x^{\omega + 1} y x^\omega = x^\omega y x^\omega = x^\omega y x^{\omega+1}$.
Besides, commutative monoids clearly satisfy the equation.
Of course, $\SG$ captures monoids not covered by~\cite{skovbjerg1997dynamic},
e.g., products of a commutative monoid and an aperiodic monoid.

The result of \cite{skovbjerg1997dynamic} uses van Emde
Boas trees~\cite{van1976design}, which we 
extend to store values in an alphabet~$\Sigma$. Fixing
an alphabet~$\Sigma$, a \emph{vEB tree} (or \emph{vEB}) is a data structure
parametrized by an integer~$n$ called its
\emph{span}, which stores a
set $X \subseteq \{1,\ldots, n\}$ and a mapping $\mu\colon X\to\Sigma$,
and supports the following
operations:
\emph{inserting} an integer $x \in \{1, \ldots, n\} \setminus X$
with a label~$\mu(x) \colonequals a$; 
\emph{retrieving} the label of~$x \in \{1, \ldots, n\}$
if $x \in X$ (or a special value if $x \notin X$);
\emph{removing} an integer~$x \in X$ and its label;
\emph{finding the next integer} of $X$ that follows
an input $x \in \{1, \ldots, n\}$ (or a special value if none exists);
and \emph{finding the previous integer}.

\begin{toappendix}
\subsection{Details on van Emde Boas Trees}

In this appendix, we give some details about the vEB data structure. It supports
the following operations:
\begin{itemize}
  \item \texttt{insert}($x$, $a$) that inserts the integer $x\in \{1,\ldots, n\}$
    in $X$ (with $x \notin X$) and sets $\mu(x) \colonequals a$;
    \item \texttt{delete}($x$) for $x \in X$ that removes the integer $x$ from
      $X$ (and from the domain of~$\mu$);
    \item \texttt{retrieve}($x$) that returns the value $\mu(x)$ if $x$ is in
      $X$ and a special value otherwise.
    \item \texttt{findPrev}($y$) for $y \in \{1, \ldots, n\}$
that returns the biggest $x\in X$ such that
$x\leq y$, and returns a special value if no such element exists, i.e., $y$ is
smaller than all elements in $X$;
    \item \texttt{findNext}($y$) that works like \texttt{findPrev} but
      returns the smallest $x \in X$ such that $x \geq y$.
\end{itemize}

The original definition~\cite{van1976design} of vEB creates a tree
with $O(n \times \log \log n)$ cells and $O(n \times \log \log n)$
initialization time. We present here an adaptation of the vEB tree
that reduces the memory usage and the initialization to $O(n)$, to guarantee
  that our preprocessing is linear.

\subparagraph*{Linear-time initialization of vEBs.}
  For this, we will reduce the predecessor problem over the range $\{1,\ldots,n\}$
  to the predecessor problem over the range $\{1,\ldots,K(n)\}$ with
$K(x)=\lceil{x/\log\log n}\rceil$.
  Note that the division here can be performed in~$O(1)$, for instance by
  filling sequentially a table for its results over $\{1, \ldots, n\}$, as the
  arguments will always be in that range.

The idea is to create an array $T$ of size $n$ that will store the mapping $\mu$
  and a vEB $V$ storing keys over the range $\{1,\ldots,K(n)\}$.

If we want to initialize our modified vEB with an empty set, we start
with $V$ empty and $T$ filled with a special value $\bot$ indicating
that the domain of the associative array is empty. We then perform the
vEB operations on this structure in the following way:
\begin{itemize}
\item  for an operation \texttt{insert}($x,a$) we set
$T[x]:=a$ and we insert into $V$ the integer $K(x)$ (if it is not
already present);
\item for an operation \texttt{retrieve}{$(x)$} we return $T[x]$;
\item for an operation \texttt{delete}($x$) we set
$T[x]:=\bot$ and if there is no $y\in K^{-1}(K(x))$ such that
$T[y]\neq \bot$ we call \texttt{delete}$(K(x))$ on $V$;
\item for  \texttt{findPrev}$(x)$, we start by looking among the
 $y\in K^{-1}(K(x))$ with $y\leq x$ and $T[y]\neq \bot$. If no such
$y$ exists, we set $p$ to the result of \texttt{findPrev}$(K(x)-1)$ on
$V$. If $p=\bot$, it means that $x$ has no predecessor and if $p$ is
set we find the predecessor of $x$ among the elements of $K^{-1}(p)$;
\item a successor query can be done in a similar fashion.
\end{itemize}

To initialize the modified vEB with a set $X$ in $O(n)$ storing the
  function $\mu$, we do something similar to performing $\card{X}$ insertions
except that we first do all the modifications in $T$ without inserting
in $V$ and afterwards we do all the insertions in $V$, making sure
that each key is inserted once. The total initialization time is thus
$O(n)$ to modify in $T$ followed by at most $K(n)$ insertions in $V$
each taking $O(\log \log n)$. Thus, the initialization time is $O(n)$ overall.

Note that, by construction, for all $x$ we have that $|K^{-1}(x)|
= \lceil \log \log n \rceil $ and thus all operations on our new vEB
run in $O(\log \log n)$.

\subsection{Proof of Theorem~\ref{thm:ubsg}}
\end{toappendix}

We can implement vEBs so that these five operations
run in $O(\log \log n)$ time in the worst case,
and so that a vEB can be constructed in linear time from an ordered list.

We use vEBs in our inductive proof to represent words with ``gaps'': a vEB represents the word obtained by concatenating
the labels of the elements of~$X$ in order.
For a semigroup~$S$
and span~$n \in \NN$, the \emph{dynamic word problem on vEBs} for~$S$ is to
maintain
a vEB $T$ of span~$n$ on alphabet~$S$ under insertions and deletions, and to
answer queries asking the evaluation in~$S$ of the word currently represented by~$T$.
As before, the complexity is the worst-case complexity of an insertion,
deletion, or query, measured as a function of the span~$n$ (which never
changes).
The data structure for this problem must be initialized during a preprocessing
phase on the initial vEB~$T$, which must run in~$O(n)$.
Note that when $T$ is empty then its evaluation is not expressible in the
semigroup~$S$: we then return a special value.

It is then clear that Theorem~\ref{thm:ubsg} follows from its
generalization to vEBs, as a word in the usual sense can be converted in linear
time to a vEB where $X = \{1, \ldots, n\}$:

\begin{toappendix}
  We will show the generalization from which Theorem~\ref{thm:ubsg} obviously
  follows, namely:
\end{toappendix}

\begin{theoremrep}
  \label{thm:ubsg2}
  Let $S$ be a semigroup in~$\SG$.
  The dynamic word problem for~$S$ on a vEB of span~$n$ is in
  $O(\log \log n)$.
\end{theoremrep}

\begin{toappendix}
We start by a preliminary remark pointing out that $\SG$ is not equal to~$\Com
  \lor \A$, as is illustrated in Figure~\ref{fig:bestiaire}:
\begin{remark}
	\label{rem:sgnotjoin}
Remark that $\SG$ contains both $\Com$ and $\A$, the variety of aperiodic monoids.
From~\cite{skovbjerg1997dynamic}, we know that both those varieties have a dynamic
word problem in $O(\log\log n)$. Hence, their \emph{join} $\Com \lor \A$, i.e.,
  the variety that they generate (which is not illustrated in
  Figure~\ref{fig:bestiaire}),
also has a word problem in $O(\log\log n)$ since the complexity of this problem
  is preserved by the operations of a variety (Proposition~\ref{prp:closure}).

However, $\SG$ is in fact not equal to this variety.
First remark that both $\Com$ and $\A$ monoids satisfy the equation
$(xyz)^{\omega+1}(xzy)^{\omega} = (xyz)^\omega (xzy)^{\omega+1}$. Therefore,
the variety they generate also satisfies this equation.
However, the syntactic monoid $M$ of the language $L \colonequals
  ((abc)^2)^*((acb)^2)^*$ does not satisfy it:
  taking $x \colonequals a$, $y \colonequals b$, $z \colonequals
  c$, the equation is not satisfied because 
  the words $(abc)^3(acb)^2$ and $(abc)^2(acb)^3$ clearly do not achieve the
  same element of~$M$ (the first one is a suffix of a word of~$L$ and the second
  is not). Still, $M$ is in $\SG$: this is simply by noticing that elements of
  the form $x^\omega y x^{\omega+1}$ always evaluate to a zero (i.e., they
  correspond to words that cannot be a factor of a word of~$L$), unless in two
  cases: (1.) $x^\omega$ corresponds to $(abc)^2$, $(bca)^2$, or $(cab)^2$, and then
  $y$ must respectively correspond to a power of $abc$, $bca$, $cab$, and the equation holds;
  or (2.) $x^\omega$ is $(acb)^2$, $(cba)^2$, or $(bac)^2$, and
  likewise $y$ must respectively correspond to a power of~$acb$, $cba$, $bac$, and the
  equation holds again.
  Thus, the
  dynamic membership problem for~$L$ is in $O(\log\log n)$
  by our results.
\end{remark}

  This remark will be refined in the next appendix section (see
  Remark~\ref{rem:r2}) to show that $\SG$ is also larger than $\Com \lor \ZG$,
  for the class $\ZG$ defined in Section~\ref{sec:zg}.

  In the rest of this appendix, we successively prove the results needed for the proof of Theorem~\ref{thm:ubsg}.
\end{toappendix}

We show this result in the rest of the section. We assume without loss of
generality 
that $S$ has a zero, as otherwise we can simply
add one. We first introduce some algebraic preliminaries, and then
present the proof, which is an induction on $\calJ$-classes of the semigroup.

\subparagraph*{Preliminaries and proof structure.}
The
\emph{$\calJ$-order}
of~$S$ is the preorder $\leq_\calJ$
on~$S$ defined by $s \leq_\calJ s'$ if $S^1 s S^1 \subseteq S^1 s' S^1$,
recalling that $S^1$ is the monoid where we add a neutral element to~$S$ if it
does not already have one.
The equivalence classes of the symmetric closure of this preorder are called \emph{$\calJ$-classes}. We 
lift the $\calJ$-order to $\calJ$-classes $C, C'$ by writing $C \leq_\calJ C'$
if $u \leq_\calJ v$ for all $u \in C$ and $u' \in C'$.
A $\calJ$-class is \emph{maximal} if it is maximal for this order.

We show Theorem~\ref{thm:ubsg2} by induction on the number of $\calJ$-classes of the
semigroup. More precisely, at every step, we consider a maximal
$\calJ$-class $C$, and remove it by reducing to the
semigroup $S \setminus C$.
Remember that the number of classes only depends on the fixed semigroup $S$, so
it is constant.
Thus, as the constant number of operations
on vEBs at each class each take time
$O(\log \log n)$, the overall bound is indeed in~$O(\log \log n)$.

The base case of the induction is that of a semigroup with a single
$\calJ$-class; from our assumption that the semigroup has a zero, that
$\calJ$-class must then consist of the zero, i.e., we have the trivial monoid
$\{0\}$, and the image is always $0$ (or undefined if the word is empty).

We now show the induction step of Theorem~\ref{thm:ubsg2}.
Take a semigroup $S$ with more than one $\calJ$-class, and 
fix a maximal $\calJ$-class $C$ of~$S$:
we know that $S \setminus C$ is not empty. What is more:

\begin{claimrep}
  \label{clm:decomp}
  For any $x, y$ of~$S$ with $xy \in C$, then $x\in C$ and $y \in
  C$.
\end{claimrep}

\begin{proof}
  Note that we have $S xy S
  \subseteq S x S$: this is because any element of $S xy S$, say $s xy t$ for
  $s, t \in S$, can be written as $s x (yt)$ so it is in $S x S$ also. Thus, we
  have $xy \leq x$. But as $C$ is a maximal $\calJ$-class, we must have $x \leq
  xy$. So $x$ and $xy$ are in the same $\calJ$-class and $x \in C$. The same reasoning shows that $y \in C$.
\end{proof}

Thus, whenever a combination of elements ``falls'' outside of the maximal
class~$C$, then it remains in~$S\setminus C$; and we can see $S \setminus C$ as
a semigroup, which still has a zero, and has strictly less $\calJ$-classes. So
we will study how to reduce to $S \setminus C$.
We now make a case disjunction depending on whether $C$ is \emph{regular}, i.e.,
whether it contains an idempotent element.

\subparagraph*{Non-regular maximal classes.}
This case is easy, because
products of
two or more elements of~$C$ are never in~$C$.
To formalize this property, for a maximal $\calJ$-class $C$ of~$S$, we call a word~$w$
on~$S^*$ \emph{pair-collapsing} for~$C$ if the product of any two adjacent
letters of~$w$ is in~$S \setminus C$. We show:

\begin{lemmarep}
  \label{lem:alwayscollapse}
  Let $C$ be a maximal $\calJ$-class. If $C$ is non-regular, then any word is
  pair-collapsing: for any $x, y \in C$, we have $x y \in S \setminus C$.
\end{lemmarep}

\begin{proof}
  We proceed by contraposition and show that if there exist $x, y \in C$ such that $x y
  \in C$, then~$C$ is regular.

  Assume that we have $x, y \in C$ such that $x y \in C$. We
  then know that $x$ and $xy$ are in the same $\calJ$-class, so $S x S = S xy
  S$. Thus, as $x \in S xy S$, there
  exists $s, t \in S$ such that $s x y t = x$. By reinjecting the left-hand side
  in itself, this equation implies that $s^i x (yt)^i = x$ for all $i \in \NN$.
  Let $\omega$ be the idempotent power of~$S$. We have $s^\omega x (yt)^\omega =
  x$. As $x \in C$, Claim~\ref{clm:decomp} implies that $s^\omega \in C$.
  Hence, $C$ contains an idempotent element, so it is regular.
\end{proof}

We can then show the following, which we will reuse for
regular maximal classes:

\begin{lemmarep}
  \label{lem:indnonreg}
  Let $S$ be a semigroup
  and let $C$ be a maximal $\calJ$-class of~$S$.
  Consider the dynamic word problem for~$S$ on vEBs of some span~$n$ where we assume that, at
  every step, the represented word is pair-collapsing for~$C$. Then that problem
  reduces to the dynamic word problem
  for~$S \setminus C$ on vEBs of span~$n$.
\end{lemmarep}

\begin{proofsketch}
  We group adjacent letters of the word into groups of $\geq 2$ letters (into a new vEB of
  the same span), so that every
  group is in~$S \setminus C$, and we can perform the evaluation with a
  structure for the dynamic word problem for $S \setminus C$.
  When handling updates, we maintain a constant bound on the size of the
  groups, without introducing singleton groups.
\end{proofsketch}

Thanks to Lemma~\ref{lem:alwayscollapse}, this allows us to settle the case of a non-regular
maximal $\calJ$-class, using the induction hypothesis to maintain the
problem on $S\setminus C$.

\begin{proof}
  We will always refer to~$w$ to mean the word on~$S$ (represented as a vEB with
  some span), and $w'$ the word on~$S \setminus C$ (represented as a vEB with
  the same span). Recall that the case where $w$ is empty is special (as the
  image may not be representable in the semigroup~$S$). The case where $w$
  contains only a single element is also special, as the result may not be in $S
  \setminus C$. Up to maintaining a count of the number of letters of~$w$, we
  can handle this special case easily: when the number of letters of~$w$ becomes
  equal to~$1$, we locate the remaining letter using the vEB, and this gives us
  immediately the evaluation of~$w$. So we assume in the sequel that $w$
  contains at least 2 elements.

  We will maintain a function $\psi$, called a \emph{position mapping},
  from the positions of~$w$ to those of~$w'$, with the following requirements:
  \begin{itemize}
    \item $\psi$ is surjective: every position of $w'$ is reached by $\psi$
    \item $\psi$ is nondecreasing
    \item $\psi$ is bounded injective: there are at most 3 positions of $w$
      mapping to any position of~$w'$
    \item $\psi$ is pair-grouping: there are at least 2 positions of $w$
      mapping to any position of~$w'$
    \item $\psi$ preserves evaluations: for every position $(i, a)$ of
      $w'$, the letter $a$ in~$S'$ is exactly the composition of
      the letters $a_1, \ldots, a_n$ of the pairs $(i_1, a_1), \ldots, (i_n,
      a_n)$ of the vEB~$w$ that are mapped to~$(i,a)$ by~$\psi$: because
      $\psi$ is nondecreasing, this is a segment of successive letters in~$w$
      (and $2 \leq n \leq 3$)
    \item For any letter $(i, a)$ in the image of $\psi$, $i$ is the index of
      the last letter of~$w$ that is mapped to~$i$ by~$\psi$
  \end{itemize}
  Thanks to the fourth condition, as the word is pair-collapsing,
  we know that the fifth condition indeed yields an evaluation in~$S \setminus
  C$.
  Also note that the position mapping (specifically conditions 1, 2 and 5)
  guarantees that $w$ and $w'$ indeed evaluate to the same monoid element,
  ensuring the correctness of the reduction: the answer to the dynamic word
  problem for~$S$ on~$w$ is the same as for~$S'$ on~$w'$. So it only remains to
  explain how we can construct $\psi$ and $w'$ from~$w$ (preprocessing), and how
  we can maintain $\psi$ and $w'$ under updates to~$w$.

  For the preprocessing,
  to initialize $\psi$, we traverse $w$ sequentially (this can be done in linear
  time if we simply assume that vEBs on span $\{1, \ldots, n\}$ always store
  their contents in an array maintained in parallel to the structure), and
  construct $w'$ and $\psi$ sequentially, by grouping the letters of~$w$ in
  groups of~$2$ or possibly~$3$ for the last group (to avoid leaving one letter
  alone). This can clearly be done in time $O(n)$ with $n$ the span of~$w$.

  To maintain $\psi$ under insertions, we find a predecessor or successor of the
  new letter in~$w$ using the vEB, we set
  the image of the new letter by~$\psi$ to be that of the predecessor or
  successor, and update the image letter in $w'$ to reflect this, in
  constant time because $\psi$ is bounded injective. The only problem is that
  now the new $\psi$ may not be bounded injective anymore, because we could have
  4 letters of~$w$ (including the new one) mapping to a position of~$w'$.
  If this happens, we split the group by inserting a letter in $w'$ for the
  two first letters of~$w$ that mapped there (put at the position of the second
  letter to satisfy condition 6) and updating the letter of~$w'$ at the
  position of the fourth letter (by condition 6) to be that of the two last
  letters. All of this can be performed with constantly many updates on~$w'$ and
  constantly many operations on the vEB~$w$.

  To maintain $\psi$ under deletions, we look at the image of the deleted letter
  by~$\psi$, we find the other letters which are in the same group (i.e.,
  constantly many neighbors, which we can find using the vEB $w$), we delete all
  letters of the group along with the group and modify~$\psi$ (preserving the
  invariant), and then we insert back the letters of the group that we did not
  intend to delete, as explained in the previous paragraph. This concludes the
  proof.
\end{proof}

\subparagraph*{Case of a regular maximal class.}
We now consider a maximal $\calJ$-class $C$ that is regular. Consider
the semigroup $C^0 \colonequals C \cup \{0\}$ for a fresh zero~$0$, i.e., the
multiplication is that of~$C$ except that $x 0 = 0 x = 0$ for all $x \in C^0$,
and that $x y = 0$ in~$C^0$ for all $x, y \in C$ for which the product $x y$
in~$S$ is not an element of~$C$.
Note that $0$ is unrelated to the zero
which $S$ was assumed to have; intuitively, the $0$ of~$C^0$ stands for
combinations of elements that are not in~$C$.
Another way to see $C^0$ is as the quotient of $S$ by the ideal $S \setminus C$, i.e., we identify all elements of $S\setminus C$ to~$0$.
By Prop.~4.35 of
Chapter~V of~\cite{mpri}, we know that $C^0$ is a so-called \emph{0-simple
semigroup}.
By the Rees-Sushkevich theorem (Theorem 3.33 
of~\cite{mpri}),
$S$ is isomorphic to some \emph{Rees
matrix semigroup with~0}. This is a semigroup $M^0(G, I, J, P)$ with $I$ and $J$
two non-empty sets, $G$ a group called the \emph{structuring group}, and $P$ a
matrix indexed by $J \times I$ having values in $G^0$. The elements of the semigroup
are the elements of $I \times G \times J$ and the element~$0$, with $x 0 = 0x = 0$ for any
element $x \in I \times G \times J$, and for $(i,g,j)$ and $(i', g', j')$ two
elements of $I \times G \times J$, their product is $0$ if $p_{j,i'} = 0$, and
$(i ,g p_{j,i'}g', j')$ otherwise.

With this representation, the idea is to collapse together the
maximal runs of consecutive elements of~$C^0$ whose product is not~$0$,
i.e., does not ``fall'' outside of~$C$. Once this is done,
the product of two elements always falls
in~$S \setminus C$, 
so we can conclude using Lemma~\ref{lem:indnonreg}.

However, we cannot do this in a naive fashion. For instance, if we insert a
letter in the vEB in the middle of such a maximal run,
we cannot
hope to split the run and know the exact group annotation of the two new runs -- this
could amount to solving a prefix-$\ZZ_d$ problem. Instead, we must now use
the fact that $S$ is in~$\SG$, and derive the consequences of the equation in terms of the
Rees-Sushkevich representation. Intuitively, the equation ensures that the 
structuring group $G$ is commutative, and that annotations in~$G$ can ``move around'' relative
to other elements in~$S$ without changing the evaluation. Formally:

\begin{claimrep}
  \label{clm:commutative}
  The structuring group~$G$ is commutative.
\end{claimrep}

\begin{proof}
  As $C$ is regular, let $(i, g, j)$ be an idempotent of~$C$. We have $(i, g, j)
  = (i, g, j)^2 = (i, g p_{j,i} g, j)$, so $g = g p_{j,i} g$, so in $G$ we have
  $g = p_{j,i}^{-1}$. This implies in particular that $p_{j,i} \neq 0$, and thus
  our idempotent is of the form $(i, p_{j,i}^{-1}, j)$.

  Let us write $H_{i,j} = \{(i, g, j) \mid g \in
  G\}$: this subset is closed under the
  semigroup operation because $p_{j,i} \neq 0$, so it is a subsemigroup of~$C^0$, and
  in fact it is a group, with neutral element $\eta \colonequals (i, p_{j,i}^{-1}, j)$, and with
  $(i, p_{j,i}^{-1} g^{-1} p_{j,i}^{-1}, j)$ the inverse
  of~$(i, g, j)$. Indeed, for any $(i,
  g, j)$ we have $(i, g, j) (i, p_{j,i}^{-1}, j) = (i, p_{j,i}^{-1}, j) (i, g,
  j) = (i, g, j)$ by the Rees law,
  so indeed $\eta = (i, p_{j,i}^{-1}, j)$ is neutral.
  Further, we have $(i,
  g, j) (i, p_{j,i}^{-1} g^{-1} p_{j,i}^{-1}, j) = (i, p_{j,i}^{-1} g^{-1}
  p_{j,i}^{-1}, j) (i, g, j) = (i, p_{j,i}^{-1}, j) = \eta$ again by the Rees law, so indeed the inverses are
  as described.

  Let us show that the group $H_{i,j}$ is commutative.
  To see why, take any $g, g' \in G$, and take $x = (i, g, j)$ and $y = (i, g',
  j)$. Let us show that $xy = yx$. Applying the equation to~$S$, we know that
  $x^{\omega+1}y x^\omega = x^\omega y x^{\omega+1}$, where $\omega$ is the
  idempotent power of~$x$.
  Now, $x^\omega = x^{2\omega}$ by definition, and in the group
  $H_{i,j}$ this implies that $x^{\omega} = \eta$, by composing the equation with $(x^\omega)^{-1}$. Thus, injecting
  this in the equation yields $x y = y x$, as claimed. Thus, $H_{i, j}$ is
  commutative.

  We now show this implies that~$G$ is a commutative group. Let us first
  note that, for any $g \in G$,  letting $1$ be the neutral element of~$G$,
  as $(i,g, j) (i, 1, j) = (i, 1, j) (i, g, j)$ because $H_{i,j}$ is
  commutative, we have $g p_{j,i} = p_{j,i} g$. Now, take any $g, g' \in G$, and
  let us show that $g g' = g' g$.
  We have $(i, g, j) (i, g', j) = (i, g p_{j,i} g', j) = (i, p_{j,i} g g', j)$ by what we
  just argued. As $H_{i,j}$ is commutative, the latter is also equal to
  $(i, g', j) (i, g, j)$, which is equal to $(i, p_{j,i} g' g, j)$. Identifying
  the two and composing by $p_{j,i}^{-1}$, which is possible because we argued
  that $p_{j,i} \neq 0$, we obtain $g g' = g' g$. So indeed
  $G$ is commutative.
\end{proof}

\begin{claimrep}
  \label{clm:oneswap}
  Let $r, s, t \in S^*$ and $(i, g, j), (i', g', j') \in I \times G \times
  J$. Write $w = r (i, g, j) s (i', g', j') t$
  and $w' = r (i, g g', j) s (i', 1, j') t$ where $1$ is the neutral element
  of~$G$. Then $\eval(w) = \eval(w')$.
\end{claimrep}

\begin{proof}
  Let $C$ be the regular maximal $\calJ$-class under consideration
  As $C$ is regular, let us consider an idempotent in~$C$. By the same reasoning
  as the first paragraph of the proof of Claim~\ref{clm:commutative}, we know
  that the idempotent is of the form $(a, p_{b,a}^{-1}, b)$ with $p_{b,a}$ being
  nonzero. What is more, like in the proof of Claim~\ref{clm:commutative}, we
  know that $H_{a,b} = \{(a, g, b) \mid g \in G\}$ is a group.

  In all equations that follow, we abuse notation and write equalities between
  words of~$S$ to mean that they evaluate to the same element (not that the
  words are the same).

  Let us now observe that:
  \[
    (i', g', j') = (i', 1, b) (a, p_{b,a}^{-1} g', b) (a, p_{b,a}^{-1}, b) (a,
    p_{b,a}^{-1}, j').\]
  Indeed, this is immediate by the law
  of the Rees semigroup with zero.

  Let us now take $x = (a, p_{b,a}^{-1} g', b)$, and take $\omega$ its idempotent power. We
  have $x^\omega = x^{2\omega}$, and as this operation happens in the group
  $H_{a,b}$ we know that $x^\omega$ is the neutral element of this group, which
  is $(a, p_{b,a}^{-1}, b)$.

  So we can rewrite the above equation as:
  \[
    (i', g', j') = (i', 1, b) x x^\omega (a, p_{b,a}^{-1}, j').
  \]
  Similarly to the first equation, we have \[
    (i, g, j) = (i, p_{b,a}^{-1}, b) (a, g, b) x^\omega (a, p_{b,a}^{-1}, j).
  \]
  So let us take:
  \[
    y = (a, p_{b,a}^{-1}, j) s (i', 1, b)
  \]
  The equation defining $\SG$, applied to~$x$ and~$y$,
  now tells us that $x^{\omega+1} y x^\omega = x^\omega y x^{\omega+1}$.

  So let us now write $w$, i.e.:
  \[
    w = r (i, p_{b,a}^{-1}, b)(a, g, b)  x^{\omega} (a, p_{b,a}^{-1}, j) s (i',
    1, b)  x x^\omega (a, p_{b,a}^{-1}, j') t
  \]
  We recognise the definition of~$y$, so we have:
  \[
    w = r (i, p_{b,a}^{-1}, b)(a, g, b)  x^{\omega} y x^{\omega+1} (a, p_{b,a}^{-1}, j') t
  \]
  Using the equation, we have:
  \[
    w = r (i, p_{b,a}^{-1}, b) (a, g, b) x^{\omega+1} y x^{\omega} (a, p_{b,a}^{-1}, j') t
  \]
  Injecting back the definition of~$y$:
  \[
    w = r (i, p_{b,a}^{-1}, b)(a, g, b)  x^{\omega+1}
    (a, p_{b,a}^{-1}, j) s (i', 1, b)
    x^{\omega} (a, p_{b,a}^{-1}, j') t
  \]
  The part between $r$ and $s$ is:
  \[
    (i, p_{b,a}^{-1}, b)(a, g, b) (a, p_{b,a}^{-1} g', b) (a, p_{b,a}^{-1}, b) (a, p_{b,a}^{-1}, j)
  \]
  which evaluates by the Rees law to $(i, g g', j)$;
  and the part between $s$ and~$t$ evaluates according to the same to $(i', 1, j')$, so we have obtained:
  \[
    w = r (i, g g', j) s (i', 1, j') t
  \]
  Which establishes that $\eval(w) = \eval(w')$ and concludes the proof.
\end{proof}

This allows us to reduce the dynamic word
problem on~$S$ to the same problem where we assume that the word is always 
pair-collapsing:

\begin{claimrep}
  The dynamic word problem for~$S$
  on vEBs (of some span~$n$) reduces to the same problem on vEBs of span~$n$ where we
  additionally require that, at every step, the represented word is
  pair-collapsing for the maximal $\calJ$-class $C$.
\end{claimrep}

\begin{proofsketch}
  We maintain a mapping where all maximal runs of word elements evaluating
  to~$C$ are collapsed to a single element, which we can evaluate following the
  Rees-Sushkevich representation. The tricky case is whenever an update breaks 
  a maximal run 
  into two parts: we cannot recover 
  the $G$-component of the annotation of each part, but we
  use Claim~\ref{clm:oneswap} to argue that we can simply put it on the left
  part without altering the evaluation in~$S$.
\end{proofsketch}

\begin{proof}
  A \emph{maximal run} of a word of~$S^*$ is a non-empty maximal contiguous
  subsequence of elements of the word whose image is in~$C$. We will maintain
  the target word $w'$ as a vEB along with a function $\psi$ from the positions
  of~$w$ to that of~$w'$, with the following requirements:

  \begin{itemize}
    \item $\psi$ is surjective: every position of $w'$ is reached by $\psi$;
    \item $\psi$ is nondecreasing;
    \item For any letter of~$w$ not part of a maximal run, i.e., an element
      of~$S \setminus C$, then the image of this letter by~$\psi$ has the same
      letter and it has has exactly this element as preimage;
    \item For any maximal run of~$w$, all its letters have the same image
      by~$\psi$, the maximal run is precisely the preimage of that element, and
      the label of this maximal run is of the form $(i, g, j)$ for some $g$,
      where $(i, g', j)$ for some~$g'$ is the actual evaluation of this maximal
      run in~$w$;
    \item For any letter $(i, a)$ in the image of $\psi$, $i$ is the index of
      the last letter of~$w$ that is mapped to~$i$ by~$\psi$
  \end{itemize}

  We additionally require that the evaluation of~$w'$ and of~$w$ in~$S$ is the
  same. It is crucial that this condition is enforced globally, not locally, as
  we intend to use Claim~\ref{clm:oneswap}.

  For the preprocessing, given the word~$w$, we process it sequentially (again
  assuming that vEBs also store their contents as an array), and we create the mapping sequentially,
  using the Rees-Sushkevich representation to determine when a maximal run ends.
  This takes time~$O(n)$.

  To handle updates on~$w$, there are several cases:
  \begin{itemize}
    \item When we insert an element of~$S \setminus C$, we check the target vEB
      to know if this insertion happens within a maximal run or not:
      \begin{itemize}
        \item If the insert is not within a maximal run, then we simply reflect
          the same change in~$w'$ and $\psi$.
        \item When we insert an element $x$ of~$S\setminus C$ within a maximal
          run, then the maximal run is broken. We use the vEB~$w$ to find the
          preceding element $(i_-, g_-, j_-)$ and $(i_+, g_+, j_+)$ of~$w$
          relative to the insertion: they are necessarily in~$C$ because we are
          within a maximal run.

          We now replace the element $(i, g, j)$ of~$w'$ which was the image of
          the maximal run by~$\psi$ by three elements: one of the form $(i, g_1,
          j_-)$ corresponding to the remaining prefix of the maximal run,
          one corresponding to the inserted element, and one of the form $(i_+,
          g_2, j)$ corresponding to the remaining suffix of the maximal
          run.  We set $g_1 \colonequals g p_{j_-, i_+}^{-1}$ and $g_2
          \colonequals 1$ for $1$ the neutral element of~$G$.

          We must argue that the evaluation of~$w'$ and of~$w$ is still the
          same. To do this, write the new $w'$ as $r (i, g_1, j_-) x (i_+, g_2,
          j) t$. Consider now $(i, g_1', j_-)$ and $(i_+, g_2', j)$, the
          respective actual evaluations of the prefix and suffix after the
          update. We know thanks to the invariant that $w$ and $w'$ had the same
          evaluation before the update, i.e., $w$ evaluates to the same in~$S$
          as $r (i, g, j) t$, and after the update $w$ evaluates to the same as
          $r (i, g_1', j_-) x (i_+, g_2', j) t$, and by Rees-Sushkevich we have
          $g = g_1' p_{j_-, i_+} g_2'$.
          Instead, our definition of $w'$ evaluates to
          $r (i, g_1 , j_-) x (i_+, 1, j) t$.
          The values $g_1'$ and $g_2'$
          are intuitively the ones that we cannot retrieve from our data
          structure. However, thanks to Claim~\ref{clm:oneswap}, as $g_1 = g
          p_{j_-, i_+}^{-1} = g_1' g_2'$ (using the commutativity of~$G$,
          Claim~\ref{clm:commutative}), we know that $w'$ as we defined it
          evaluates to the correct value. Thus, the invariant is preserved.
      \end{itemize}
    \item When we delete an element of~$S\setminus C$:
      \begin{itemize}
        \item If the deletion does not connect together two maximal runs (i.e.,
          it is preceded and succeeded in~$w'$ by elements that are not both
          in~$C$, or that are both in~$C$ but whose composition in~$C$
          yields~$0$), then we simply reflect the change in~$w'$;
        \item When we delete an element of~$S \setminus C$ and connect together
          two maximal runs, then we update~$w'$ to delete the element and then
          delete the two elements $(i, g, j)$ and $(i', g', j')$ corresponding
          to the two runs by an element $(i, g p_{j,i'} g', j')$ corresponding
          to the new run.
      \end{itemize}
    \item When we insert an element $(i, g, j)$ of~$C$,
      we check $w'$ to distinguish many possible cases:
      \begin{itemize}
        \item If this happens within a maximal run (as ascertained using the vEB
          structure on~$w'$), then let $(i_-, g_-, j_-)$
          and $(i_+, g_+, j_+)$ be respectively the preceding and succeeding
          elements in~$w$, obtained from the vEB of~$w$: they are both in~$C$.
          We update the $G$-annotation of the image by~$\psi$ of this maximal
          run to add $p_{j_-,
          i_+}^{-1}$, i.e., the inverse of the Rees matrix element obtained
          between the preceding and succeeding element: call this step (*). Now:
          \begin{itemize}
            \item If both $p_{j_-,i}$ and $p_{j,i_+}$ are nonzero then we simply
              update the $G$-annotation again to add $p_{j_-, i}$ and
              $p_{j, i_+}$.
            \item If both $p_{j_-,i}$ and $p_{j,i_+}$ are zero, then we break
              the maximal run in three parts: the part before the insertion, the
              insertion which is a maximal run of its own, and the part after
              the insertion. We replace the element $(i',g',j')$ in
              $w'$ (with $g'$ already modified by step (*))
              that corresponded to the whole run by three elements $(i',
              g', j_-)$, $(i, g, j)$, and $(i_+, 1, j')$.
              The preservation of the global invariant is again by
              Claim~\ref{clm:oneswap}.
            \item The two cases where exactly one of $p_{j_-,i}$ and $p_{j,i_+}$
              is zero and the other is nonzero are analogous to the above.
          \end{itemize}
        \item If this does not happen within a maximal run:
          \begin{itemize}
            \item If the preceding
          element is in $S\setminus C$ or is an element of~$C$ which combined
          with~$(i,g,j)$ gives a zero according to the Rees matrix, and the same
          is true of the succeeding element, then $(i,g,j)$ is a new maximal run
          of its own. We insert it as-is in~$w'$
        \item If the preceding
          element is in~$S\setminus C$ or an element of~$C$ which combined
          with~$(i,g,j)$ gives a zero, but the next element $(i_+,g_+,j_+)$ is in~$C$ and
          combining $(i,g,j)$ with it does not give a zero (i.e., $p_{j,i_+}$ is
          nonzero), then we are
          extending the maximal run that follows.
              We modify its $G$-annotation to add $p_{j,i_+} g$, we remove the element
          in~$w'$ corresponding to that consecutive run and add it back to
              the new end position of the extended maximal run, but changed so
              that its $I$-index becomes~$I$.
        \item The case of an insertion extending the maximal run that precedes
          is analogous except that we can simply update the element in~$w'$
              without having to move it, and change its~$J$-index and not
              $I$-index.
        \item If both the preceding element $(i_-, g_-, j_-)$ and succeeding
          elements $(i_+, g_+, j_+)$ are in~$C$
          and neither gives a zero together with the newly inserted element,
              then by our assumption that we are not within a maximal run, it
              must be that $p_{j_-,i_+} = 0$ and we are merging together the two maximal runs of these
              elements. We reflect this in the additional structure, in the
              $G$-annotation (adding the two Rees matrix terms
              $p_{j_-,i}$ and $p_{j,i_+}$ which are nonzero by assumption, in
              addition to~$g$),
              removing in~$w'$ the element for the first run, and copying
              its~$I$-index and group information to the element for the
              original run (now the element for the second run)
          \end{itemize}
      \end{itemize}
    \item When we remove an element $(i,g,j)$ of~$C$, we distinguish between several
      cases:
      \begin{itemize}
        \item Removal within a maximal run which does not change the run
          endpoints, i.e., the preceding and succeeding elements do not combine
          to a zero in the Rees matrix semigroup. Then we simply update the
          $G$-annotation to add the inverse of the two Rees matrix
          entries that are no longer realized and add~$g^{-1}$, and add the new nonzero Rees
          matrix entry which is realized.
        \item Removal within a maximal run which breaks up the run into two
          nonempty runs at the deletion point. We insert the information for
          these two runs, putting the group information of the split run on the
          first run (along with $g^{-1}$ and the inverse $p$-term for the two
          elements that are no longer adjacent) and again use
          Claim~\ref{clm:oneswap} to argue for correctness.
        \item Removal of the first letter of a maximal run. We update the
          $G$-annotation with the inverse of the Rees matrix entry
          which is no longer realized, and update the $I$-index of the end of the
          run.
        \item Removal of the last letter of a maximal run.  This is like the
          previous case except we change the $J$-index instead of the $I$-index,
          and we must move the element in~$w'$ corresponding to the run to
          sit at the new ending position of the run.
        \item Removal that eliminates a singleton maximal run. We simply delete
          it from~$w'$. \qedhere
    \end{itemize}
    \end{itemize}
\end{proof}

This claim together with Lemma~\ref{lem:indnonreg} implies that
the dynamic word problem for~$S$ reduces to the same problem for~$S \setminus
C$, for which we use the induction hypothesis. This establishes the induction
step and concludes the proof of Theorem~\ref{thm:ubsg}.

\section{Dynamic Word Problem for Monoids in~$\ZG$}
\label{sec:zg}
We pursue our study of the dynamic word problem for monoids with
the class $\ZG$, introduced
in~\cite{auinger2000semigroups} 
and defined by the equation $x^{\omega+1} y = y x^{\omega+1}$ for all $x, y$.
This asserts that
elements of the form $x^{\omega +1}$, which are the ones belonging to a
\emph{subgroup} of the monoid, are \emph{central}, i.e.,
commute with all other elements. By the equations, and recalling that
$x^{\omega+1} = x^\omega x^{\omega+1}$, clearly
$\ZG \subseteq \SG$.
In this section, we show an $O(1)$ upper bound on the dynamic word problem for
monoids in~$\ZG$, and a conditional lower bound 
for any monoid not in~$\ZG$.

\begin{toappendix}
We start by another preliminary remark on the relationship between $\SG$, $\ZG$,
  and $\A$ (see Figure~\ref{fig:bestiaire}):
\begin{remark}
  \label{rem:r2}
	Following Remark~\ref{rem:sgnotjoin}, we could be tempted to claim that
	$\SG$ is equal to the variety generated by $\ZG$ and $\A$. However, as
        monoids in $\ZG$
	also satisfy the equation $(xyz)^{\omega+1} (xzy)^\omega = (xyz)^\omega (xzy)^{\omega+1}$, 
	the same argument shows that it is not the case.
\end{remark}
We then provide the omitted proofs for the results in the main text.
  \subsection{Proof of Upper Bound Results}
\end{toappendix}

\subparagraph*{Upper bound.}
Recall the result on commutative monoids from~\cite{skovbjerg1997dynamic}:

\begin{theoremrep}[\cite{skovbjerg1997dynamic}]
  \label{thm:ubcom}
  The dynamic word problem for any commutative monoid is in $O(1)$.
\end{theoremrep}

\begin{proof}
  This result is proved in~\cite{skovbjerg1997dynamic}, but we re-prove it for
  completeness.
  Let $M$ be the monoid and let $n$ be the length of the input
  word~$w$. We simply precompute a table of the powers $x^i$ for
  every $x \in M$ and
  every $0 \leq i \leq n$. Doing so by successive composition
  takes time $O(n)$ (recalling that $M$ is fixed).
  We also count the number
  of occurrences $n_x$ of each monoid element~$x$ in~$w$. Now, the vector $\vec{n}$ can easily be
  maintained in $O(1)$ under letter substitution updates on~$w$: when overwriting $x$ by $y$,
  we increment $n_y$ and decrement $n_x$. Thanks to commutativity, the result of
  the evaluation is then $\prod_{x \in M} x^{n_x}$ which we can evaluate in
  constant time thanks to the precomputed values. This establishes the result.

  In fact, we can even avoid precomputing the tables by recalling that the
  subsets of $\NN^{|\Sigma|}$ of the vectors evaluating to each of the monoid
  elements are \emph{recognizable subsets} which by~\cite[Proposition
  9]{Choffrut06} are defined by a finite number of congruence and threshold
  conditions. These can be checked in~$O(1)$ without using precomputed powers.
\end{proof}

Our goal is to generalize it to the following result:

\begin{theorem}
  \label{thm:ubzg}
  The dynamic word problem for any monoid in~$\ZG$ is in $O(1)$.
\end{theorem}

This generalizes Theorem~\ref{thm:ubcom} (as
commutative monoids are clearly in~$\ZG$)
and covers other monoids,
e.g., the monoid $M = \{1, a, b, ab, 0\}$ with $a^2=b^2=ba=0$, where it intuitively
suffices to track 
the position of~$a$'s and~$b$'s and compare them if there is only one of each.

We now prove Theorem~\ref{thm:ubzg}.
A semigroup $S$ is \emph{nilpotent} if it has a zero and there exists
$k > 0$ such that $S^k = \{0\}$, i.e., all products of $k$ elements 
are equal to~$0$.
Alternatively~\cite[Chapter~X, Section~4]{mpri}, $S$ is nilpotent
iff it satisfies the equation $x^\omega y = y x^\omega = x^\omega$.
We then consider the monoids of the form $S^1$ where $S$ is nilpotent --
an example of this is the monoid $M$ described above.
The variety generated by such monoids is called $\MNil$ and
was studied
by Straubing~\cite{Straubing82}. We can show:
\begin{propositionrep}
	\label{prop:nil}
	For any nilpotent 
        $S$, the dynamic word problem for
        $S^1$ is in~$O(1)$.
\end{propositionrep}

\begin{proofsketch}
  We maintain a (non-sorted) doubly-linked list $L$ of the positions of the
  word~$w$ that contain
  a non-neutral element.
  As $S$
  is nilpotent, the evaluation of~$w$ is~$0$ unless constantly many
  non-neutral letters remain, which we can then find in~$O(1)$ with~$L$.
\end{proofsketch}

\begin{proof}
        Let $k > 0$ be the constant integer such that $S^k = 0$. Given
        a word~$w \in S^*$, we prepare in linear time from~$w$ a doubly-linked
        list $L$ containing the positions of~$w$ having an element which is not the
        identity of~$S^1$, along with a table $T$ of size~$\card{w}$ where the
        $i$-th cell contains a pointer to the list element in~$L$ representing
        the $i$-th element if some exists, and a dummy value otherwise. We can
        construct $L$ in linear time.

        We will maintain the invariant that $L$ contains all positions of~$w$
        containing a non-neutral element (note that we do not assume that $L$
        is in sorted order), and that $T$ contains one pointer per element of~$L$
        leading to the cell corresponding to that element in~$L$.

        We can easily use~$L$ to determine the image of the current word
        in~$S^1$. We first check if $L$ contains $\geq k$ elements, which can be
        done in time $O(k)$, hence $O(1)$, by navigating the list.
        If this is the case, then as $S^k = 0$, we know that the word evaluates
        to~$0$. Otherwise, we know the $< k$ non-neutral elements of~$w$, and we
        can evaluate their product in~$O(1)$ to know the answer.

        We now explain how to maintain $L$ in constant time per update. When an
        update replaces $1$ by~$1$, or a non-neutral element by a
        non-neutral element, then we do nothing. When an update replaces a
        neutral element by a non-neutral element at position~$i$, then we
        add~$i$ to~$L$, and let $T[i]$ be a pointer to the new list item, in time
        $O(1)$. When an update replaces a non-neutral element by a neutral
        element at position~$i$, we use $T[i]$ to find the element for~$i$
        in~$L$, remove it in time $O(1)$, and erase $T[i]$, all in~$O(1)$. This
        concludes the proof.
\end{proof}

In~\cite{localityarxiv} we show that $\ZG$ is generated by such monoids $S^1$ 
and by commutative monoids:
\begin{proposition}[Corollary~3.6 of~\cite{localityarxiv}]
  \label{prp:zgcarac}
	The variety $\ZG$ is generated
  by commutative monoids and monoids of the form $S^1$ for
$S$ a nilpotent semigroup.
\end{proposition}

In view of Theorem~\ref{thm:ubcom} and Proposition~\ref{prop:nil}, the dynamic
word problem is in $O(1)$ for the semigroups that generate the variety $\ZG$
(Proposition~\ref{prp:zgcarac}). Theorem~\ref{thm:ubzg} then follows from
Proposition~\ref{prp:closure}.

\begin{toappendix}
  \subsection{Proof of Lower Bound Results: Theorem~\ref{thm:lbzg}}
\end{toappendix}

\subparagraph*{Lower bound.}
We now show a conditional lower bound on the dynamic word problem for
monoids outside of~$\ZG$. To do this, we will reduce from the prefix-$U_1$
problem:

\begin{theorem}
  \label{thm:lbzg}
	For any monoid~$M$ in $\SG \setminus \ZG$, the
        prefix-$U_1$ problem reduces to the dynamic word problem for~$M$.
\end{theorem}

\begin{proofsketch}
  We consider the variety $\ZE$~\cite{almeida1987implicit} of monoids whose idempotents are central,
i.e., the variety defined by the equation $x^\omega y = y x^\omega$.
  We show that $\ZG = \SG\cap \ZE$. We then show that, for any monoid not
  in~$\ZE$, we can reduce from the prefix-$U_1$ problem by encoding the elements
  $0$ and $1$ of~$U_1$ using carefully chosen elements of the monoid.
\end{proofsketch}

Using Conjecture~\ref{con:cpq}, and together with Theorem~\ref{thm:lbsg} for the
monoids not in~$\SG$, this implies a conditional super-constant lower
bound 
for monoids outside~$\ZG$.

\begin{toappendix}
We prove Theorem~\ref{thm:lbzg} in the rest of this section. To do this,
let us introduce $\ZE$ as the variety of monoids whose idempotents are central,
i.e., the variety defined by the equation $x^\omega y = y x^\omega$. Note that
$\ZG \subseteq \ZE$, and more precisely:

\begin{claimrep}
  \label{clm:ze}
  We have: $\ZG = \SG\cap \ZE$.
\end{claimrep}

\begin{proof}
  Let $M$ be in $\ZG$. We first show that $M$ is in~$\ZE$. Consider arbitrary
  elements~$x$ and~$y$. We have $x^\omega = (x^\omega)^{\omega+1}$ by
  definition of~$x^\omega$. Thus,
  $x^\omega y = (x^\omega)^{\omega+1} y = y (x^\omega)^{\omega+1}= y x^\omega$.
  Thus, $M$ is in~$\ZE$. Furthermore, as $\ZG \subseteq \SG$, clearly $M$ is
  in~$\SG$.

  Let $M$ be in $\SG\cap \ZE$. Then
\begin{align*}
x^{\omega+1}y & = x^{\omega+1} x^\omega y & (\text{By definition of } x^\omega)\\
              & = x^{\omega+1} y x^\omega & (\text{By }M\in \ZE)\\
              & = x^\omega y x^{\omega+1} & (\text{By }M\in \SG)\\
              & = y x^\omega x^{\omega+1} & (\text{By }M\in \ZE)\\
              & = y x^{\omega+1}&
\end{align*}
This concludes the proof of the claim.
\end{proof}

  We can now conclude the proof of Theorem~\ref{thm:lbzg}:
\begin{proof}
  Let $M$ be a monoid in $\SG \setminus \ZG$. By Claim~\ref{clm:ze}, we know
  that $M$ is not in~$\ZE$, so there exist
$x,y\in M$ such that $x^\omega y \neq y x^\omega$.

  Notice that we cannot have
  both $x^\omega y = x^\omega y x^\omega$ and $y x^\omega = x^\omega y
  x^\omega$, so one of these equations must be false.
  Without loss of generality, we can assume that $y x^\omega \neq x^\omega y
  x^\omega$. Indeed, if we have $x^\omega y \neq x^\omega y x^\omega$ instead,
  then we can show instead the lower bound for the \emph{reversal} $M^t$ defined
  by $x \cdot_{M^t} y = y \cdot x$ with $\cdot_{M^t}$ and $\cdot_{M}$ the
  internal laws of $M^t$ and $M$ respectively. Now, we can obviously reduce from the
  dynamic word problem for~$M^t$ to the same problem for~$M$ by 
  reversing the input word and performing the
  updates at the mirror position. Thus, it suffices to consider the case where
  $y x^\omega \neq x^\omega y x^\omega$

We now show that that any solution to the dynamic membership problem
of $M$ can be used to solve the prefix-$U_1$ problem. To do so, we consider a
  word $w$ on~$\{0, 1\}$ of length~$n$, and encode it as a word $w'$ of length
  $2n+2$. All letters of~$w'$ are
  the neutral element $e$ of $M$, except that $w_{2n+2}' = x^\omega$ and
  $w_{2i}' =
  x^\omega$ whenever $w_i = 0$. This can be done in linear time during the
  preprocessing. Now,
  any update that writes $0$ or $1$ in~$w_i$ is done by writing respectively $x^\omega$ or~$e$
  to $w'_{2i}$.

  Now, to perform at prefix-$U_1$ query with argument $j$, we write $w_{2j+1}
  \colonequals y$, use the dynamic word problem data structure to get the
  evaluation of the word, then write back $w_{2j+1} \colonequals e$.
  The evaluation result, after removing neutral elements,
  is $x^{k\omega} y x^{k'\omega}$ where $k$ is the number of $0$'s in the prefix
  of length~$j$ in~$w$, and $k'$ is the number of $0$'s in the rest of~$w$,
  which is $\geq 1$.
  Because $x^\omega=x^{2\omega}$ this is equivalent to $x^\omega y x^\omega$
  when the prefix contained a $0$, or to $y x^\omega$ when it did not. We have
  shown that these two elements are different, so we can indeed recover the
  answer to the prefix query, concluding the proof.
\end{proof}

\end{toappendix}

\section{Dynamic Word Problem for Semigroups}
\label{sec:semigroup}
We have classified the complexity of the dynamic word problem
for monoids: it is in $O(\log \log n)$ for monoids in~$\SG$, in $O(1)$ for
monoids in~$\ZG$, in $\Omega(\log n / \log \log n)$ 
for monoids not in~$\SG$, and non-constant for monoids
not in~$\ZG$ conditionally to Conjecture~\ref{con:cpq}.
In this section, we extend our results from monoids to \emph{semigroups}.

\subparagraph*{Submonoids and local monoids.}
A \emph{submonoid} of a semigroup $S$ is a subset of the semigroup which is stable
under its composition law and is a monoid.
We first notice via Proposition~\ref{prp:closure} that a semigroup that
contains a hard
submonoid is also hard:
\begin{claimrep}
	\label{clm:reduce-subsg}
       The dynamic word problem for any submonoid of a semigroup~$S$ reduces to the same problem for $S$.
\end{claimrep}

\begin{proof}
  This is immediate by Proposition~\ref{prp:closure} (and also intuitively): we
  simply solve the problem with a structure for~$S$ but where we only use
  elements of the submonoid.
\end{proof}

We will investigate if studying the submonoids of a semigroup suffices to
understand the complexity of its dynamic word problem. To do so, we focus
on a certain kind of submonoids: the \emph{local monoids}. A submonoid $N$ of $S$ is a \emph{local monoid} if there exists an idempotent
element $e$ of~$S$ such that
$N=eSe$, i.e., $N$ is the set of elements that can be written as $ese$ for some
$s \in S$. The point of local monoids is that they are maximal
in the sense that every submonoid $T$ of $S$ is a
submonoid of a local monoid:
indeed, taking $1$ the neutral element of~$T$, all elements of~$T$ can be
written as $1T1 \subseteq 1S1$ and $1S1$ is a local monoid.
For $\V$ a variety of monoids, we denote by $\LV$ the variety of semigroups
such that all local monoids are in $\V$.
As we explained, this is equivalent to imposing that all submonoids are in $\V$
(since varieties are closed under the submonoid operation).

\subparagraph*{Case of~$\SG$.}
We now revisit our results on monoids to extend them to semigroups, starting
with~$\SG$. 
We denote by $\LSG$ the variety of semigroups whose local monoids are in~$\SG$.
We show that a semigroup where all local monoids are
in~$\SG$ must itself be in~$\SG$:

\begin{claimrep}
	 We have $\LSG = \SG$ as varieties of semigroups.
\end{claimrep}
\begin{proof}
  Clearly $\SG \subseteq \LSG$. Let $S$ be a semigroup of~$\LSG$.
	For all elements $x,y$ of $S$, letting
	$e\colonequals x^\omega$, we have that the local monoid $N=eSe$ is in $\SG$.
        Now, since $x'=exe$ and $y'=eye$ are in $N$,
        we have $(x')^{\omega+1} y' (x')^\omega = (x')^\omega y' (x')^{\omega+1}$. Furthermore,
        $(x')^{\omega+1}=x^{\omega+1}$ and $(x')^\omega = x^\omega$.
        Thus $x^{\omega+1} y x^\omega = x^\omega y x^{\omega+1}$.
\end{proof}

Semigroups in~$\SG$ are already covered by our upper
bound (Theorem~\ref{thm:ubsg}), and semigroups not in~$\LSG$ 
have a submonoid not in~$\SG$, so we can use
Claim~\ref{clm:reduce-subsg} and Theorem~\ref{thm:lbsg}. Hence:

\begin{corollary}
  \label{cor:sg}
	Let $S$ be a semigroup. If $S$ is in $\SG$, then the dynamic word
        problem for~$S$ is in $O(\log \log n)$.
        Otherwise, the dynamic word problem for~$S$ is in $\Omega(\log n /
        \log \log n)$.
\end{corollary}

\subparagraph*{Case of~$\ZG$.}
The variety $\ZG$ is not equal to $\LZG$.
For instance, let $S$ be the syntactic semigroup of $a^*b^*$, that is the semigroup
 $\{a, b, ab, 0\}$ defined with $a^2=a$, $b^2=b$ and $ba=0$.
It is not in $\ZG$, since $a$ and $b$ are idempotents that do not commute. However,
its local monoids are either trivial or $U_1$, so they are all in $\ZG$, showing that
this semigroup is in~$\LZG$.
Still, we can extend our characterization from monoids to semigroups up to
studying~$\LZG$:

\begin{theorem}
	\label{thm:lzgeq}
        Let $S$ be a semigroup. If $S$ is in $\LZG$, then the dynamic word
        problem for~$S$ is in $O(1)$.
        Otherwise, unless prefix-$U_1$ is in $O(1)$, the dynamic word problem
        for~$S$ is not in~$O(1)$.
\end{theorem}

The second part of the claim is by
Claim~\ref{clm:reduce-subsg} and Theorem~\ref{thm:lbsg}, but
the first part is much trickier. We use a characterization of~$\LZG$ 
as a \emph{semidirect
product} $\ZG*\D$, which follows from a very technical
\emph{locality result} on~$\ZG$~\cite{localityarxiv}, and then design an
algorithm for the dynamic word problem for semigroups in $\ZG*\D$.
We prove Theorem~\ref{thm:lzgeq} in the rest of this
section.

Given two semigroups $S$ and $T$, a \emph{semigroup action} of $S$ on $T$ is
defined by a map $\act\colon S\times T \to T$ such that
$\act(s_1, \act(s_2, t))= \act(s_1s_2, t)$ and $\act(s, t_1t_2)=\act(s,t_1) \act(s,t_2)$.
We then define the  product $\circ_\act$ on the set
$T\times S$ as follows: for all $s_1, s_2$ in $S$
and $t_1, t_2$ in $T$, we have: $(t_1, s_1)\circ_\act(t_2, s_2) \colonequals (t_1 \act(s_1, t_2), s_1 s_2).$
The set $T\times S$ equipped with the product $\circ_\act$ is a semigroup called the
\emph{semidirect product} of $S$ by $T$, denoted $T\circ_\act S$.

We say that a semigroup $D$ is \emph{definite} if there exists an integer $k \in \NN$ such that
for all $y, x_1, \ldots, x_k$ in $D$, we have $yx_1\cdots x_k = x_1\cdots x_k$.
Alternatively,
a semigroup is definite iff it satisfies the equation $y x^\omega = x^\omega$~\cite[Proposition 2.2]{straubing1985finite}
for all $x,y$ in $D$.
In particular, every nilpotent semigroup is definite.
We write $\D$ for the variety of definite semigroups.

Our alternative definition of $\LZG$ will be the variety of semigroups $\ZG*\D$ 
generated by semigroups that are the semidirect product of a
$\ZG$  monoid by a definite semigroup.

The variety $\ZG*\D$ of semigroups is not immediately related to the
variety~$\LZG$ defined above.
One can easily show that $\ZG*\D \subseteq \LZG$, but the other direction is
much more challenging to establish.
We show this as a so-called \emph{locality theorem}, which we defer to a
separate paper~\cite{localityarxiv} because it uses different tools and is of possible independent interest:

\begin{theorem}[\cite{localityarxiv}, Theorem~1.1]
	We have: $\ZG*\D = \LZG$.
\end{theorem}

To conclude the proof of Theorem~\ref{thm:lzgeq}, by the locality theorem above,
it suffices to solve the dynamic word problem for semigroups in~$\ZG*\D$. By
Proposition~\ref{prp:closure}, it suffices to consider the semigroups that
generate the variety. We do this below, establishing
Theorem~\ref{thm:lzgeq}:

\begin{propositionrep}
  Let $S$ be a definite semigroup, let $T$ be a semigroup of~$\ZG$, and let
  $\act$ be a semigroup action of~$S$ on~$T$. The
  dynamic word problem for the semigroup $T\circ_\act S$ reduces to the same problem
  for~$T$.
\end{propositionrep}

\begin{proofsketch}
  We express the direct product of the letters of the input word as a product involving elements of~$T$ and prefix
  sums of elements of~$S$, which we can maintain in~$O(1)$.
\end{proofsketch}

\begin{proof}
  Given a word $w = (t_1, s_1),\ldots, (t_n,s_n)$ of $(T\circ_\act S)^*$, we note
  that the second component of its evaluation is $s_1 \cdots s_n$. As $S$ is
  definite, we can compute this and maintain it in constant time, simply by looking
  at the $k$ last elements, for $k$ the integer witnessing that $S$ is definite.

  As for the first component, it can be shown to evaluate to the following word
  of~$T^*$:
    \[t_1 \cdot\act(s_1, t_2)\cdot \act(s_1s_2, t_3) \cdots \act(s_1s_2\cdots
    s_{n-1}, t_n)\]
  We initialize a structure for the dynamic word problem on~$T$ with this
  word~$w'$
  to obtain the first component. Now, when the word~$w$ is updated at
  position~$i$, we perform the updates on this word~$w'$ by changing $t_i$ in
  cell~$i$ (a single update), and by changing the $k$ cumulative sums of~$s$ that
  have changed, i.e., the products in the first component of~$\act$ in the up to
  $k$ cells starting at cell~$i$: this amounts to~$k$ updates, and for each of
  them the new cumulative sum can be computed by looking at the $k$ last
  elements, so this is a constant time computation and a constant number of
  operations.
\end{proof}

\section{Dynamic Word Problem for Languages}
\label{sec:languages}
\begin{toappendix}
  \label{apx:languages}
\end{toappendix}
We now turn to the dynamic membership problem for regular
languages, and show Theorem~\ref{thm:mainres} using the three previous sections
and some extra algebraic results.

\subparagraph*{Connection to the dynamic word problem.}
A regular language $L$ is \emph{recognized} by a finite monoid if there exists a morphism
$\eta\colon\Sigma^*\to M$ such that $L=\eta^{-1}(\eta(L))$.
The \emph{syntactic congruence} of~$L$ is the equivalence relation on~$\Sigma^*$
where $u, v \in \Sigma^*$ are equivalent iff, for each $r, t \in \Sigma^*$,
either $rut \in L$ and $rvt \in L$, or $rut \notin L$ and $rvt \notin L$.
The \emph{syntactic monoid}~$M$ of~$L$ is the quotient of $\Sigma^*$ by the
syntactic congruence of~$L$, and the \emph{syntactic morphism}
is the morphism mapping~$\Sigma^*$ to~$M$; the syntactic morphism witnesses that
the syntactic monoid recognizes~$L$.

The dynamic membership problem for a language clearly reduces to the dynamic word
problem for its syntactic monoid. However, the converse is not true: the
language $L\colonequals (aa)^*ba^*$ has a syntactic monoid $M$ that can be shown
to be outside of~$\SG$, but we can solve dynamic membership for~$L$ in~$O(1)$ 
by counting the $b$'s at even and odd positions. Intuitively, $M$
has a neutral element~$1$ so that the dynamic word problem for~$M$ has a
reduction from prefix-$\ZZ_2$, but $1$ is not achieved by a letter of the
alphabet so dynamic membership for~$L$ is easier.

We extend our results to languages using the notion of \emph{stable
semigroup}~\cite{Barrington92,Chaubard06}. This allows us to remove the neutral element (as it is a
semigroup not a monoid) and ensures that all semigroup elements 
can be achieved by subwords of some constant length (the \emph{stability
index}).

Formally, let $L$ be a regular language and $\eta\colon\Sigma^* \to M$ its syntactic morphism.
The \emph{powerset} of~$M$ is the monoid
whose elements are subsets of $M$ and for $E,F\subseteq M$, the product $EF$ is
$\{xy\mid x\in E, y\in F\}$.
The \emph{stability index} of~$L$ is the idempotent power~$s$ of $\eta(\Sigma)$ in the
powerset monoid.
Intuitively, this choice of $s$ ensures that, for any two words
$w_1, w_2 \in \Sigma^s$, the value $\eta(w_1 w_2)$ of their concatenation in
the monoid can be achieved by another word of~$\Sigma^s$, i.e., $\eta(w_1 w_2) =
\eta(w)$ for some $w \in \Sigma^s$.
Then $\eta(\Sigma^s)$ is a subsemigroup of $M$, because 
$(\eta(\Sigma^s))^2 = \eta(\Sigma^s)$: we call it 
the \emph{stable semigroup} of~$L$.
For any class of semigroups $\V$, we denote by $\QV$ the class of languages
whose stable semigroup is in $\V$.

\subparagraph*{Upper bounds.}
We first show that the dynamic membership problem for a regular language
reduces more specifically to the dynamic word problem for its \emph{stable
semigroup}:
\begin{propositionrep}
	\label{prop:stable}
	Let $L$ be a regular language. The dynamic membership problem
     for~$L$ reduces to the dynamic word
	problem for the stable semigroup of~$L$.
\end{propositionrep}

\begin{proofsketch}
  We partition the word of~$L$ into chunks of size~$s$ (plus one of size $\leq s$)
  for $s$ the stability index, and feed them to the data structure for the stable semigroup of~$L$.
\end{proofsketch}

\begin{proof}
  Let $L$ be a regular language with $S$ its stable semigroup.  We reduce the dynamic
membership problem for~$L$ to the dynamic word problem for~$S$ as
follows.  For any word $u\in \Sigma^*$, we decompose $u$ into $u=u_1
u_2 \cdots u_n v$ with each $u_i$ of length $s$ and $v$ having length
$\leq s$, where $s$ denotes the stability index. The image by $\eta$ of
this decomposition gives a word of $S^* M$. We use a maintenance
scheme for the word of~$S^*$, and propagate the updates on~$u$ to
updates of this word in $O(1)$ in the expected way: we compute the
corresponding letter of the word of~$S^*$ by dividing the position of
the update by a constant, we look a constant number of neighboring
positions to find the entire $u_i$ in the decomposition of~$u$, and
evaluate~$\eta$ to compute the new image.  Note that the case of
updates to~$v$ is immediate as~$v$ has constant size.  We can then use
the structure for the dynamic word problem on~$S^*$ to obtain the
image of~$u_1 \cdots u_n$ in the stable semigroup, which we can
compute together with~$v$.
\end{proof}

By Corollary~\ref{cor:sg} and Theorem~\ref{thm:lzgeq}, this implies that
languages in~$\QSG$ (resp., in $\QLZG$) have a dynamic membership problem in
$O(\log \log n)$ (resp., in $O(1)$).

\subparagraph*{Lower bounds.}
We now show that languages whose stable semigroup is not in $\LV$ admit a
reduction from the dynamic word problem of a monoid of~$\V$.

\begin{toappendix}
  We now prove Proposition~\ref{prp:lblang} in the rest of the appendix.
\end{toappendix}

\begin{propositionrep}
  \label{prp:lblang}
	Let $\V$ be a variety of monoids and let $L$ be a regular language not
        in~$\QLV$. There 
	is a monoid not in $\V$ whose dynamic word problem reduces 
        to the dynamic membership problem for~$L$.
\end{propositionrep}

\begin{proofsketch}
  If $L$ is not in $\QLV$, then its stable semigroup contains a submonoid $M$ not
  in~$\V$, and all elements can be achieved by a block of~$s$ letters for $s$
  the stability index.
\end{proofsketch}

\begin{toappendix}
  We first
  establish a general claim formalizing the connection between the syntactic
  monoid and the dynamic word problem:

\begin{proposition}
  \label{prp:connection}
	Let $L$ be a regular language and $\eta$ its syntactic morphism.
        The dynamic membership problem for~$L$ is equivalent under constant-time
        reductions to the dynamic word problem for its syntactic monoid $M$ where we
        require that we only use elements of~$\eta(\Sigma)$.
\end{proposition}

\begin{proof}
	Clearly, a maintenance scheme for the dynamic word problem of~$M$ provides a maintenance scheme for $L$
	as it is sufficient to check that the resulting element belongs to
        $\eta(L)$, and we will indeed only use elements of~$\eta(\Sigma)$.

	In the other direction, we use the well-known fact on the syntactic
        morphism that,
        for all $m \in M$, the language $L_m \colonequals \eta^{-1}(m)$
	is in the Boolean algebra closed under the quotient operator generated
        by $L$.
        In other words, any language $L_m$ can be expressed using $L$, Boolean
        operations, and the quotient operator; none of these operations changes
        the alphabet.
	Thus, we can reduce the dynamic word problem for~$L$ to the problem
        of checking if the image is in~$L_m$ for every possible choice of $m \in
        M$. Now, each of the $L_m$ reduces to~$L$, without changing the
        alphabet, by the analogue of
        Proposition~\ref{prp:closure} on languages, which can apply to the
        quotient operator and Boolean operators. So we can
        solve the dynamic word problem for~$M$ under the
        assumption that we gave, by building data structures for these~$L_m$:
        this uses the assumption that we are only using letters
        from~$\eta(\Sigma)$.
\end{proof}

  We can now prove Proposition~\ref{prp:lblang}:

  \begin{proof}[Proof of Proposition~\ref{prp:lblang}]
Assume $L$ is not in $\QLV$. Let $\eta$ be its syntactic morphism, $s$ be its stability index
and $S$ its stable semigroup.
Then, by definition of $L$, there exists a submonoid $N$ of $S$ which is not in $\V$.
By definition of the stable semigroup, there exists a mapping $\psi$ from $N$ to $A^s$
such that for all $n\in N$, we have $\eta(\psi(n)) = n$, for $s$ the stability index.
We simulate the dynamic word problem for $N$ by inserting for each update at position $i$
the corresponding word at position $s\times i$. We can perform the evaluation in $N$ by evaluating
the image by $\eta$ of the resulting word. We know that evaluating the image by
    $\eta$ reduces
  in constant time to $L$ by Proposition~\ref{prp:connection}, where we use the
  fact that the resulting word (formed of the blocks of size~$s$) only consists
  of letters from the original alphabet.
\end{proof}

\end{toappendix}

Again by Corollary~\ref{cor:sg} and Theorem~\ref{thm:lzgeq}, we deduce that
languages outside of~$\QSG$ have complexity at least $\Omega(\log n / \log \log
n)$. Further,
assuming Conjecture~\ref{con:cpq},
and languages outside of $\QLZG$ do not have complexity $O(1)$.

\section{Extensions, Problem Variants, and Future Work}
\label{sec:extensions}
We have presented our results on the dynamic word problem for monoids and
semigroups, and on the dynamic membership problem for regular languages. We
conclude the paper by a discussion of problems for further study. We first
discuss the question of \emph{intermediate complexities} between $O(1)$ and
$O(\log \log n)$. We then study the complexity of \emph{deciding which case
applies} as a function of the target language, semigroup, or monoid. We last
explore the issue of \emph{handling insertions and deletions} on the input word,
and of supporting \emph{infix queries}.

\subparagraph*{Intermediate complexities.}
Our $O(\log \log n)$ upper bound in Theorem~\ref{thm:ubsg} and its
variants may not be tight.
Still, we can identify a language $L_{U_2}$ in~$\QSG\setminus\QLZG$ for
which we show that the dynamic membership problem is in $\Theta(\log \log n)$ (even allowing
randomization and allowing a probably correct answer), because
the prefix-$U_2$ problem reduces to it.

We can also identify a language of~$\QSG\setminus\QLZG$ that reduces to
prefix-$U_1$ and so can be solved in expected $O(\sqrt{\log \log n})$.
This shows that languages in $\QSG\setminus\QLZG$ have different complexity
regimes, at least when allowing randomization.

\begin{propositionrep}
There is a language $L_{U_2}$ in~$\QSG\setminus\QLZG$ which is equivalent to
  prefix-$U_2$ under constant-time reductions,
and a language $L_{U_1}$ in $\QSG\setminus\QLZG$ which is equivalent to prefix-$U_1$
  under constant-time reductions.
\end{propositionrep}

\begin{proof}
  We show each claim of the statement separately.

\subparagraph*{Claim on $L_{U_2}$.}
Let $L_{U_2}$ be the regular language over the alphabet $\Sigma=\{a,b,c,x\}$
that contains all words where there is only one $x$ and the closest preceding non-$c$
letter exists and is a~$b$, i.e., $L_{U_2}$ can be defined through the regular expression
  $(a+b+c)^*bc^*x(a+b+c)^*$. This language is in $\SG$ and hence can be
maintained in $O(\log \log n)$.

  We first design a reduction from
  the prefix-$U_2$ problem to the dynamic membership problem for~$L_{U_2}$, which implies
  that there is an $\Omega(\log \log n)$ lower bound on that problem.

Let $w$ be the word over the alphabet $1,a,b$ to maintain for
prefix-$U_2$. The idea of the reduction consists in encoding
$w=w_1 \cdots w_n$ into $h(w)=h(w_1) \cdots h(w_n)$ where
  $h(1)=c$, $h(a)=a$ and $h(b)=b$. Clearly $h(w)\not\in L_{U_2}$ as
there is no~$x$ in the word.

Now, to handle a prefix-$U_2$ query with parameter $k$, we look at the
$k$-th letter of $w$. If it is $a$ or $b$ we can answer immediately
with the value of $w_k$. In the remaining case of $w_k=1$ we set
$w'_k$ to~$x$.
  If we now have $w'\in L_{U_2}$ it means that the answer of
the prefix query is $b$, otherwise it means that the answer is either
$a$ or $1$. To distinguish the two cases we first look at
$w_1$, if $w_1\neq 1$ then the answer is $a$ otherwise we set
  $w'_1$ to $b$, if $w'$ belongs to $L_{U_2}$ then the answer is $1$
otherwise it is $a$. In all these cases after answering the query we
restore $w'$ to its previous state.

  We then design
  a reduction from the dynamic membership problem for~$L_{U_2}$ \emph{to} the prefix-$U_2$ problem. We simply encode
  a word of $\{a, b, c, x\}$ by writing $a$ as~$a$, $b$ as~$b$, $c$
  as~$1$, and $x$ as~$1$. We also maintain a doubly-linked
  list as in the proof of Proposition~\ref{prop:nil} to store all occurrences
  of~$x$. Now, whenever the number of $x$'s is different from $1$ then the word
  does not belong to the language. Otherwise, we use $L$ to find the position of
  the one~$x$. The word is then in the language iff the closest preceding
  non-$c$ letter exists and is a~$b$, which we can determine with the
  prefix-$U_2$ data structure by testing for the corresponding prefix and
  checking if the answer is~$b$.

\subparagraph*{Claim on $L_{U_1}$.}
Let $L_{U_1}$ be the regular language over the alphabet
$\Sigma=\{a,c,x\}$ defined by the expression $c^*x(a+c)^*$, i.e., there is
  exactly one~$x$ and it precedes all~$a$'s.

  We first design a reduction from the dynamic membership problem
  for~$L_{U_1}$ to the prefix-$U_1$ problem.
  Indeed, we translate an input word $w$
  on $\{a, c, x\}$ to a word on~$\{0, 1\}$ by writing $0$ for $a$ and $1$ for
  $a$ and $x$, and maintain this under updates. We also maintain a doubly-linked
  list as in the proof of Proposition~\ref{prop:nil} to store all occurrences
  of~$x$. Now, whenever the number of $x$'s is different from $1$ then the word
  does not belong to the language. When the number becomes equal to~$1$, the
  list $L$ allows us to know its position, and a prefix query on the
  prefix-$U_1$ structure allows us to know if there is an $a$ before this~$x$ (in
  which case the word is not in the language) or if there is none (in which case
  it is).

  We then design a reduction \emph{from} the prefix-$U_1$ problem to
  the dynamic membership problem for~$L_{U_1}$. Indeed, we
  simply encode $1$ as~$c$ and $0$ as~$a$. When a prefix query arrives for a
  prefix of length~$i$, we check if the $i$-th letter is a~$0$, in which case we
  return~$0$, and otherwise we update the word on~$\{a, b, x\}$ to insert an~$x$
  at that position. Now, the resulting word belongs to the language iff there is
  no preceding~$0$.
\end{proof}

\subparagraph*{Deciding which case applies.}
A natural question about our results is the question of
efficiently identifying, given a regular language, which of the cases of
Theorem~\ref{thm:mainres} applies, or in particular of determining, given an
input monoid or semigroup, if
it is in~$\SG$, or in~$\ZG$.
This depends on how the input is represented.
If we are given a monoid explicitly (as a table of its operations), then the
equations of $\ZG$ and of $\SG$ can be checked in polynomial time.
If the monoid is represented more concisely as the transition monoid of some
automaton, then the verification can be performed in PSPACE. We do not know if
the problems are PSPACE-hard, though this seems likely at least for~$\SG$
because of its proximity to aperiodic monoids~\cite{Cho91}.
We leave open the precise complexities of this task, in particular for the
$\mathbf{L}$ and $\mathbf{Q}$ operators.

\subparagraph*{Handling insertions and deletions.}
Another natural question is to handle insertion and
deletion updates, i.e., \emph{inserting} letter $a$ at position $k$ 
transforms the word $w_1 \cdots w_{k-1} w_k \cdots w_n $ into 
$w_1 \cdots w_{k-1} a w_k \cdots w_n $, and \emph{deleting} at position
$k$ does the opposite. Any regular language
can be maintained under such updates in $O(\log n)$ using a Fenwick
tree, but it makes the problem much harder for most
languages. For example, if the alphabet has two letters $a$ and~$b$, just testing if the word that we maintain contains an $a$
requires $\Omega(\log n / \log \log n)$ by~\cite{cstheorylist}.
This is why we do not study such updates in this work. Interestingly, notice
that our algorithm in Theorem~\ref{thm:ubsg2} supports insertions and deletions
on words represented as
vEBs, but the semantics are different (they use explicit positions in a fixed range).

\subparagraph*{Infix queries.}
A natural extension of dynamic membership for a regular language~$L$ is
the \emph{dynamic infix membership problem}, where we can query any \emph{infix} of the
word (identified by its endpoints) to ask
whether it is in~$L$. The $O(\log n / \log \log n)$ algorithm of
Theorem~\ref{thm:ub_ln_lln} supports this,
and so can the $O(\log \log n)$ algorithm
of~\cite{skovbjerg1997dynamic} for group-free monoids.
However, the infix problem can be harder. Consider for instance the language $L_2$ on $\{a, b\}$ of words with an even number
of~$a$'s. Dynamic membership has complexity
$O(1)$ because $L_2$ is commutative, but infix queries (or even prefix queries) require $\Omega(\log n / \log
\log n)$ as prefix-$\ZZ_2$ reduces to it.

We leave open the study of the complexity of the infix problem.
We note, however, that this
problem for a language $L$ can be studied as the dynamic membership problem for
a
regular language defined from~$L$. So our results cover
the infix problem via this transformation; we leave to future work a
characterization of the resulting classes.

\begin{claimrep}
  For any fixed regular language $L$, the dynamic infix membership problem
  is equivalent up to constant-time reductions to the dynamic membership problem
  for the language $\Sigma^* x L x \Sigma^*$ where~$x$ is a fresh letter.
\end{claimrep}

\begin{proof}
  Let $L' \colonequals \Sigma^* x L x \Sigma^*$, where $x$ is a fresh letter not
  in~$\Sigma$.
  We first explain how to use a dynamic membership data structure for~$L'$ to
  solve the dynamic infix membership problem for~$L$. Given a word
  over~$\Sigma^*$, we initialize the structure for~$L'$ on this word where we
  add one letter to the beginning of the word and one letter to the end of the
  word, say~$a$. When updates are performed on the word for~$L$, we
  perform them in the data structure by offsetting them by 1. Now, whenever we
  receive an infix query for a subword, we perform two updates
  on~$L'$ to replace the characters immediately before and after the subword
  by~$x$, check if the resulting word belongs to~$L'$ using the data structure,
  and undo these two operations to put back the correct characters. It is clear
  that the modified word is in~$L'$ iff the infix is in~$L$. Note that the
  addition of the two fixed letters at the beginning and end of the word
  guarantee that the characters immediately before and after the subword are
  indeed defined.

  Second, we explain how to use a dynamic infix membership data structure
  for~$L$ to solve the dynamic membership problem for~$L'$. Given a word over
  $(\Sigma \cup \{x\})^*$, we initialize the structure for~$L$ by replacing
  every occurrence of $x$ by some arbitrary character of~$\Sigma$. We also
  prepare a doubly-linked list, like in the proof of Proposition~\ref{prop:nil},
  storing all occurrences of the letter~$x$, as well as a pointer from positions
  containing~$x$ to the element of the doubly-linked list storing this copy
  of~$x$. All of this can be performed as part of the preprocessing.

  Given updates to the word on~$L'$, we maintain the doubly-linked list and
  pointers, and we replicate these updates on the word on~$L$, except that
  occurrences of~$x$ are replaced by some arbitrary character of~$\Sigma$.

  Now, to know if the current word belongs to~$L'$ or not, first observe that
  this is never the case if the number of occurrences of~$x$ is different
  from~$2$, which we can check in constant time using the doubly-linked list. If
  there are exactly two~$x$'s, we know their position, and can perform an infix
  query on the data structure for~$L$ for this infix. The current word is
  in~$L'$ iff this query returns true. Thus, by performing this infix query
  after every update to the word on~$L'$ (whenever the current word contains
  exactly two $x$'s), we obtain the desired information. This concludes the
  proof.
\end{proof}

\subparagraph*{Other open questions.}
A natural question for future work would be to study the complexity of
our problems in weaker models, e.g., pointer
machines~\cite{TARJAN1979110}, or machines with counters. One could
also extend our study to languages that are not regular, e.g., generalizing 
bounds on maintaining the language of
well-parenthesized strings (\cite[Proposition~1]{husfeldt1998hardness} and
\cite{frandsen1995dynamic}).

\bibliography{bib}

  \pagebreak

\begin{table}[ht]
  \caption{Summary of the main classes of monoids and semigroups used in the
  paper}
  \label{tab:classes}
\begin{tabularx}{\linewidth}{lXr@{\hspace{0.2em}}lr}
  \toprule
  {\bfseries Class} & {\bfseries Description} &\multicolumn{2}{c}{\bfseries
  Equation} & {\bfseries References} \\
\midrule

        $\ZE$	& Monoids/semigroups with central idempotents			  		&	$x^\omega y 				=$ 	&	$ 			y x^\omega		$		& \cite{almeida1987implicit} \\
        $\ZG$	& Monoids/semigroups with central groups  			  		&	$x^{\omega+1} y 			=$	&	$  			y x^{\omega+1}	$	&\cite{auinger2000semigroups}	\\
        $\SG$ 	& Monoids/semigroups with swappable groups			 	 		&	$x^{\omega+1} y x^\omega	=$ 	&	$x^\omega	y x^{\omega+1}	$	& \cite{azevedo1990join,skovbjerg1997dynamic} 	\\
        $\A$	& Aperiodic semigroups/monoids 			 		 		&	$x^{\omega+1} 				=$	&	$			x^\omega		$	& \cite{mpri}	\\
        $\Com$	& Commutative semigroups/monoids				 	 		&	$xy							=$	&	$ 			yx				$	& \cite{mpri}	\\
        $\Nil$	& Nilpotent semigroups				 	 		&	$x^\omega y 				=$	&	$			y x^\omega		$	& \cite{mpri}	\\
        $\MNil$	& Monoids dividing a nilpotent semigroup	 		&									&				& \cite{Straubing82}						\\
        $\D$	& Definite semigroups					 	 		& 	$y x^\omega					=$	&	$ 			x^\omega		$	& \cite{straubing1985finite}	\\
                                                \bottomrule
\end{tabularx}
\end{table}

  \begin{figure}[h!]
\includegraphics{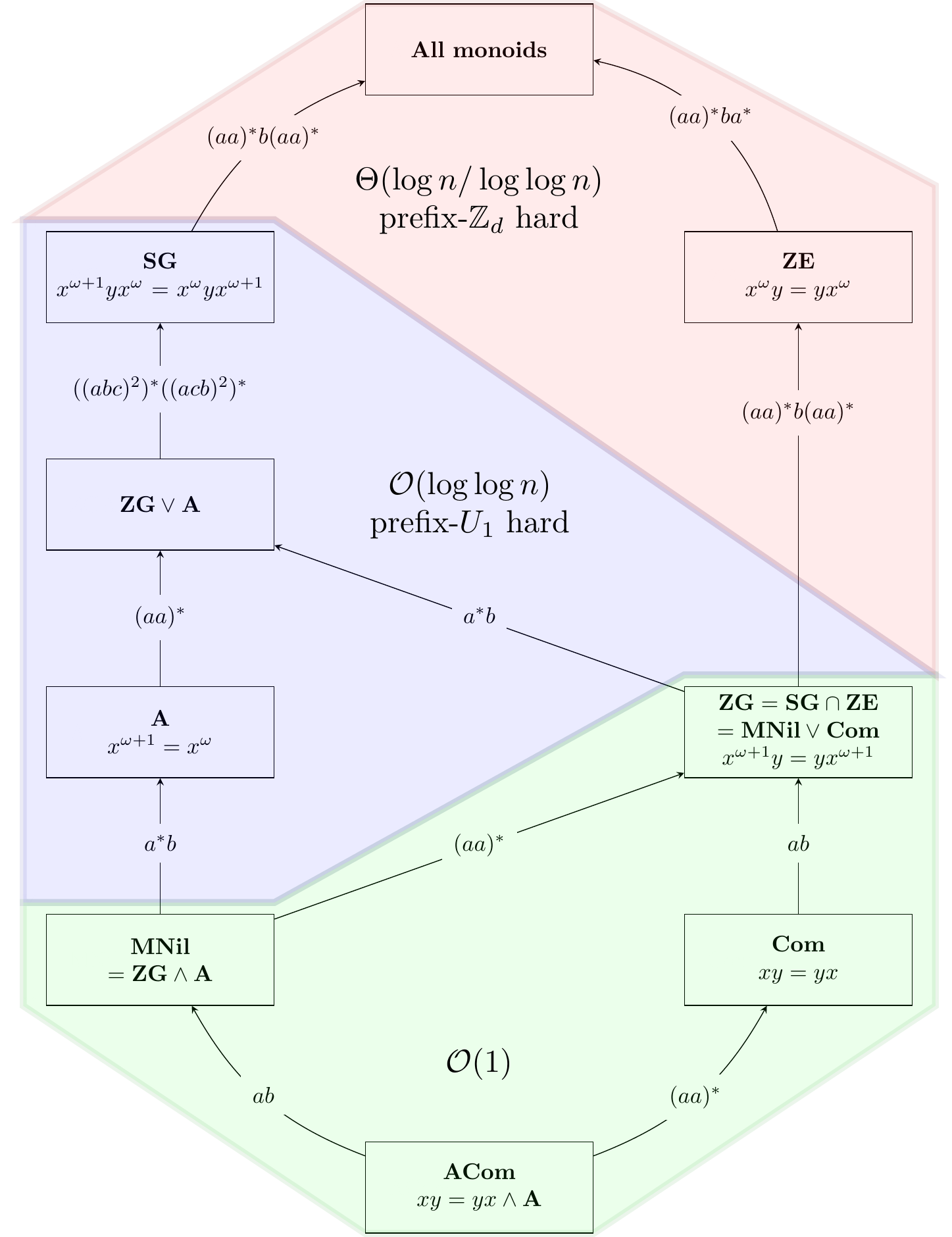}
\caption{Complexity of the dynamic word problem for common classes of monoids.
  Arrows denote inclusion and are labeled with languages (with an implicit
  neutral letter~$e$) whose syntactic monoids
  separate the classes. The classes $\ZG$ and $\SG$ are maximal for the $O(1)$
  region and $O(\log \log n)$ region respectively.}
  \label{fig:bestiaire}
\end{figure}

\end{document}